\documentclass[lettersize,journal]{IEEEtran}
\IEEEoverridecommandlockouts
\usepackage{cite}
\usepackage{amsmath,amssymb,amsfonts}
\usepackage{algorithmic}
\usepackage{graphicx}
\usepackage{textcomp}
\usepackage{xcolor}
\usepackage{float}
\usepackage{grffile}
\usepackage[linesnumbered,ruled]{algorithm2e}
\usepackage{amsfonts,amsmath,amsthm,amssymb,graphicx,epstopdf,cite,url,subfigure}
\usepackage{psfrag,amssymb,amsmath,pifont,cite,graphics,graphicx,epsfig,subfigure,amscd,helvet,multirow,enumerate}
\usepackage[colorlinks, citecolor=blue,linkcolor=blue]{hyperref}
\usepackage{color}
\usepackage{epstopdf}
\usepackage{stfloats}
\usepackage{accents}
\allowdisplaybreaks[1]
\theoremstyle{definition}
\newtheorem{theorem}{Theorem}
\theoremstyle{definition}
\newtheorem{lemma}{Lemma}
\theoremstyle{definition}
\newtheorem{definition}{Definition}

\newtheorem{remark}{Remark}
\def\BibTeX{{\rm B\kern-.05em{\sc i\kern-.025em b}\kern-.08em
    T\kern-.1667em\lower.7ex\hbox{E}\kern-.125emX}}

\begin{document}

\title{Sub-Nyquist Sampling for Reaching Theoretical Minimal Sampling Rate Boundary}

\author{%
	\IEEEauthorblockN{Dong Xiao and Jian Wang}

	\IEEEauthorblockA{School of Data Science,
		Fudan University, Shanghai 200433, China\\
		E-mail: dxiao24@m.fudan.edu.cn, jian\_wang@fudan.edu.cn}

}

\markboth{Journal of \LaTeX\ Class Files,~Vol.~14, No.~8, August~2021}%
{Shell \MakeLowercase{\textit{et al.}}: A Sample Article Using IEEEtran.cls for IEEE Journals}


\maketitle

\begin{abstract}
	Wideband spectrum sensing motivates sub-Nyquist sampling architectures that exploit spectral sparsity, yet in blind scenarios where subband locations are unknown, existing schemes require sampling rates at least twice the theoretical minimum. To this end, we propose a dual-frequency aliasing wideband converter (DAWC), which partitions the multiband spectrum into non-uniform frequency intervals and selectively samples only a subset of them, requiring no prior knowledge of subband locations. We demonstrate that under mild conditions on the signal and the system, DAWC achieves perfect subband localization and waveform reconstruction at the theoretical minimum rate. Moreover, we introduce an innovative side-information-aided subspace pursuit (MSSP) algorithm exploiting the common support structure inherent in the signal column submatrices for exact recovery of the spectrum support set. Based on the restricted isometry property (RIP), we provide stable recovery guarantees for MSSP in the presence of noise. Numerical simulations show that the proposed scheme achieves superior spectrum recovery accuracy compared to state-of-the-art methods.
\end{abstract}


\section{Introduction}

\IEEEPARstart{M}{odern} wireless and sensing systems, ranging from 5G millimeter-wave communications~\cite{rappaport2013millimeter} and ultra-wideband radar~\cite{baraniuk2007compressive} to non-terrestrial networks~\cite{giordani2020nonterrestrial}, routinely span bandwidths from hundreds of megahertz to several gigahertz. Acquiring such signals at the Nyquist rate demands analog-to-digital converters (ADCs) whose speed, cost, and power consumption often exceed practical limits~\cite{sun2012wideband,walden1999analog,mishali2011sub}. Fortunately, wideband spectra are typically sparse, with only a small fraction of the available bandwidth actively occupied at any given time. Sub-Nyquist sampling architectures exploit this sparsity through compressed sensing (CS) techniques, recovering the signal of interest from far fewer measurements than the Nyquist rate would require~\cite{Donoho2006,Chen2006,Ariananda,Mishali,Lopez-Parrado}.

In many practical systems, however, the subband locations are unknown a priori. In this blind sensing scenario, the minimum sampling rate required to reconstruct the signal is theoretically bounded by twice the total occupied bandwidth~\cite{mishali2009blind}. To approach this minimum rate, conventional sub-Nyquist sampling schemes (CSSS)~\cite{Mishali,Venkataramani,Tropp2010,Moon,Xampling,Haque} typically employ multiple parallel channels, each performing spectrum aliasing, frequency-domain filtering, and discrete-time sampling. This processing chain partitions the wideband spectrum into uniform intervals, so that each sub-Nyquist measurement amounts to a weighted combination of spectral content across the intervals.

While uniform spectral partitioning provides a feasible framework, a fundamental inefficiency arises from the blind nature of the signals. If an interval is occupied by a subband, the CSSS requires at least two sampling channels to guarantee perfect recovery of the spectral components within that interval~\cite{Yue,song2022approaching,jang2018intentional}. The sampling cost therefore scales with the bandwidth of activated intervals rather than the actual spectral occupancy. This cost coincides with the theoretical minimum only when every subband aligns exactly with an interval boundary. In practice, however, subband locations bear no fixed relation to the interval boundaries, so a single subband frequently straddles adjacent intervals, activating more bandwidth than it actually occupies. Consequently, the sampling rate of a practical CSSS inevitably exceeds the theoretical minimum~\cite{mishali2009blind,song2022approaching}.

To close this gap, several strategies have been explored. One line of work~\cite{song2022approaching,jang2018intentional} refines the spectral partitioning to mitigate subband straddling, yet since the partition remains uniform, subbands whose edges do not coincide with the refined grid still activate excess intervals. Another~\cite{sun2012wideband,xiong2015compressed,molazadeh2019novel} employs multi-rate architectures that model the spectrum as a generic sparse vector, circumventing the fixed-partition constraint at the cost of more stringent recovery conditions due to a single-measurement-vector formulation. Other approaches reduce the sampling burden by exploiting second-order statistical priors such as power spectral sparsity and cyclostationarity~\cite{cohen2014unified,cohen2017cyclostationary,vasavada2020sparse,yang2020fast}, or by first estimating coarse parameters such as sparsity order and then adapting the sampling configuration accordingly~\cite{wang2012sparsity,ma2023revisiting}. The former relies on statistical structures that practical signals do not always exhibit, while the latter is sensitive to estimation errors in low signal-to-noise ratio or rapidly changing spectral conditions.

Beyond the sampling architecture itself, reconstruction performance depends critically on how well the recovery algorithm exploits the structure of the sparse spectrum. Joint sparsity methods~\cite{tropp2006algorithms,Kim,Jian,CIte47} improve recovery by enforcing a common support across multiple measurement vectors. Building on this idea, block sparse methods~\cite{Eldar,Mukhopadhyay,eldar2010block,yuan2006model,zhang2013extension} go a step further by grouping adjacent intervals to capture the contiguity of active subbands. Nevertheless, both families of methods constrain only which intervals are active, without exploiting the structure within each active interval. Due to the spectral continuity of subbands, sub-portions within a single active interval tend to share highly similar sparsity patterns. This intra-interval sparsity similarity provides additional structural information that has not been exploited by existing methods.

In this paper, for a general class of multiband signals whose subband locations are unknown, we propose a sub-Nyquist sampling scheme called dual-frequency aliasing wideband converter (DAWC) to close the gap between the practical sampling rate and the theoretical minimum. The main idea of DAWC is to enclose the subband support within known bounds via the non-uniform partition in the frequency-domain by shifting the spectrum with two fundamental frequencies. Subsequently, DAWC resamples the signal based on the estimated subband positions to achieve waveform reconstruction. We provide the spectral analysis of the sampling processes and derive that under mild conditions on the signal and system parameters, DAWC can achieve perfect signal reconstruction at the minimum rate.

To recover the support set from the compressive measurements, we propose the multiple side-information-aided subspace pursuit (MSSP) algorithm, which exploits the intra-interval sparsity similarity of the spectrum. The core of MSSP is to decompose the compressive measurement model into multiple column-wise subproblems and share mutual support information across them to enhance recovery performance. We further establish that if the sensing matrix satisfies the restricted isometry property (RIP)~\cite{candes2005decoding} with a restricted isometry constant, MSSP guarantees robust recovery of any $s$-sparse signal within a finite number of iterations.

The rest of this paper is organized as follows. In Section II, we introduce notations, the background and the lower bound of sampling rate in the CSSS. In Section III, we introduce the DAWC and drive its input-output relationship. We provide the conditions between the signal and the parameters of DAWC for attaining the minimum rate. In Section IV, we provide a useful observation regarding the steps of MSSP and elucidate the role of the common sparsity structure within the framework. Furthermore, we provide a RIP-based theoretical analysis of the recovery performance of MSSP. In Section V, we study the empirical performance of the DAWC and MSSP and concluding remarks are given in Section VI.


\section{BACKGROUND AND PROBLEM FORMULATION}
\subsection{Notations}
We summarize notations used throughout this paper.
The symbol $\mathrm{j}$ stands for the imaginary unit.
The notations $\lfloor \cdot \rfloor$, $\lceil \cdot \rceil$,
and $|\cdot|$ denote the floor function, ceiling function,
and cardinality of a set, respectively.
Given a universal set $\Omega$, the complement of $S$ is
$S^c = \Omega \setminus S$, and for $n \in \mathbb{Z}^+$,
$[n] = \{1, \ldots, n\}$.
For $\mathbf{X} \in \mathbb{C}^{L \times N}$ and a set $S$,
$\mathbf{X}_{S}$ (or $\mathbf{X}^{S}$) denotes the submatrix
with columns (or rows) indexed by $S$; the $i$-th row,
$j$-th column, and $(i,j)$-th entry are written as
$\mathbf{X}_{i,:}$, $\mathbf{X}_{:,j}$, and $\mathbf{X}_{i,j}$.
The row support $\operatorname{supp}(\mathbf{X})$ is the set of
non-zero row indices.
For $\mathbf{X} \in \mathbb{C}^{L \times N}$ and
$\mathbf{x} \in \mathbb{C}^{L \times 1}$,
$\|\mathbf{X}\|_{F}$, $\|\mathbf{X}\|_{2 \to 2}$,
$\|\mathbf{X}\|_{0} \triangleq |\operatorname{supp}(\mathbf{X})|$,
and $\|\mathbf{x}\|_{2}$ denote the Frobenius, spectral,
$\ell_0$, and $\ell_2$ norm, respectively.
For $\mathbf{A} \in \mathbb{C}^{p \times L}$ with
$\mathbf{A}_{S}$ having full column rank,
$\mathbf{A}^{\dagger}_{S}$, $\mathbf{A}^{H}_{S}$, and
$\mathbf{A}^{\top}_{S}$ denote the Moore--Penrose
pseudo-inverse, conjugate transpose, and transpose.
The orthogonal projections onto the column space of
$\mathbf{A}_{S}$ and its complement are
$\mathbf{P}_{S} = \mathbf{A}_{S}\mathbf{A}_{S}^{\dagger}$
and $\mathbf{P}_{S}^{\perp} = \mathbf{I}_{p} - \mathbf{P}_{S}$.
The Frobenius inner product is
$\langle \mathbf{A}, \mathbf{X} \rangle
	= \operatorname{tr}(\mathbf{A}^{H}\mathbf{X})$.
\subsection{Conventional Sub-Nyquist Sampling Scheme}\label{sub-nyqu}

Consider a multiband signal $x(t)$ supported on $[0, T]$, whose Fourier transform is given by $X(f) = \int_0^T x(t)\, e^{-\mathrm{j}2\pi ft}\, \mathrm{d}t$. The signal is bandlimited to $\mathcal{F}_{\mathrm{nyq}} = [-f^{\max}, f^{\max})$, i.e., $X(f) = 0$ for $f \notin \mathcal{F}_{\mathrm{nyq}}$, with Nyquist frequency $f_{\mathrm{nyq}} = 2f^{\max}$. The frequency support of $x(t)$ consists of $N_{\mathrm{sig}}$ disjoint subbands~\cite{Haque,Yue,song2022approaching,cohen2014sub}:
\begin{equation}
	\mathcal{T} = {\bigcup}_{j=1}^{N_{\mathrm{sig}}} [f_j^{\min},\, f_j^{\max}),
\end{equation}
where $-f^{\max} \leq f_1^{\min} < f_1^{\max} < \cdots < f_{N_{\mathrm{sig}}}^{\min} < f_{N_{\mathrm{sig}}}^{\max}$ and the bandwidth of the $j$-th subband is $B_j = f_j^{\max} - f_j^{\min}$.

The conventional sub-Nyquist sampling schemes (CSSS), such as the multi-coset sampler (MCS)~\cite{Venkataramani} and the modulated wideband converter (MWC)~\cite{Mishali} mainly include two stages: (i) spectrum aliasing and (ii) acquisition.
\begin{itemize}
	\item Stage (i) involves shifting $X(f)$ by $\frac{mf_{\mathrm{nyq}}}{L}$, where $m \in \{ -\left\lfloor \frac{L}{2} \right\rfloor, \dots, \left\lfloor \frac{L-1}{2} \right\rfloor \}$ for $L \in \mathbb{Z}^+$, via downsampling or modulation in $p$ parallel channels. In the $m$-th channel, the aliased spectrum is a weighted sum of shifted replicas:
	      \begin{equation}\label{eq:NEWWW1}
		      \mathcal{Y}_m(f) = \sum_{j=-\left\lfloor \frac{L}{2} \right\rfloor}^{\left\lfloor \frac{L-1}{2} \right\rfloor}a_{m,j}\mathcal{X}_j(f) ,
	      \end{equation}
	      where $a_{m,j}\in \mathbb{C}$ is a weighting coefficient determined by the specific architecture, while $\mathcal{X}_j(f) \triangleq X(f + j\,\frac{f_{\mathrm{nyq}}}{L})$ represents the $j$-th frequency-shifted replica.
	\item In stage (ii), the output of each channel from stage (i) is then processed by a low-pass filter (LPF) with a passband $\mathcal{F}_{\mathrm{LPF}}\triangleq [0,\mathcal{W}_{\mathrm{LPF}}]$, where $\mathcal{W}_{\mathrm{LPF}} \geq \frac{f_{\mathrm{nyq}}}{L}$. Subsequently, the filtered output is sampled at a frequency of $f_s \geq \mathcal{W}_{\mathrm{LPF}}$. Without loss of generality, let $\mathcal{W}_{\mathrm{LPF}} = \frac{f_{\mathrm{nyq}}}{L} = f_s$, the discrete sampled spectrum is then given by:
	      \begin{align}
		      \mathcal{Y}_m(f) & = \sum_{j=-\left\lfloor \frac{L}{2} \right\rfloor}^{\left\lfloor \frac{L-1}{2} \right\rfloor} a_{m,j} \mathcal{X}_j(f), \nonumber \\
		      f                & \in  \left\{ \frac{k \mathcal{W}_{\mathrm{LPF}}}{\mathcal{N}-1} \mid k \in \mathbb{Z},\ 0 \leq k \leq \mathcal{N}-1 \right\}
	      \end{align}
	      where $\mathcal{N} = Tf_s$ denotes the length of $\mathcal{Y}_m(f)$.
\end{itemize}

The input-output relationship of the system can be expressed as a CS model $\mathcal{Y} = \mathcal{A}\mathcal{X}$, where
\begin{align}
	\mathcal{Y} & = \begin{bmatrix}
		                \mathcal{Y}_1(f_0), \cdots, \mathcal{Y}_1(f_{\mathcal{N}-1}) \\
		                \vdots                                                       \\
		                \mathcal{Y}_p(f_0), \cdots, \mathcal{Y}_p(f_{\mathcal{N}-1})
	                \end{bmatrix} \in \mathbb{C}^{p\times \mathcal{N}},                                                                                  \\
	\mathcal{X} & = \begin{bmatrix}
		                \mathcal{X}_{-\left\lfloor \frac{L}{2} \right\rfloor}(f_0), \cdots, \mathcal{X}_{-\left\lfloor \frac{L}{2} \right\rfloor}(f_{\mathcal{N}-1}) \\
		                \vdots                                                                                                                                       \\
		                \mathcal{X}_{\left\lfloor \frac{L-1}{2} \right\rfloor}(f_0), \cdots, \mathcal{X}_{\left\lfloor \frac{L-1}{2} \right\rfloor}(f_{\mathcal{N}-1})
	                \end{bmatrix} \in \mathbb{C}^{L\times \mathcal{N}},\label{eq:ne444w}
\end{align}
and $f_k = \dfrac{k \mathcal{W}_{\mathrm{LPF}}}{\mathcal{N}-1} $ for $k \in \{ 0,\dots,\mathcal{N}-1 \}$. The sensing matrix $\mathcal{A} \in \mathbb{C}^{p \times L}$ is defined with its $(m,j)$-th entry given by $\mathcal{A}_{m,j} = a_{m,j-\left\lfloor \frac{L}{2} \right\rfloor -1 }$. Under a certain error constraint $\epsilon > 0$, the aim of CSSS is to reconstruct the signal matrix as follows:
\begin{equation}\label{quen:1}
	\mathcal{X}^* = \arg\min_{\mathcal{X}} \|\mathcal{X}\|_0, \quad \|\mathcal{Y} - \mathcal{A}\mathcal{X}\|_F < \epsilon.
\end{equation}

\subsection{The Lower Bound of Sampling Rate}\label{Section_II_C}
Recall that the CSSS employs $p$ parallel channels, each operating at a sampling rate $f_s$, yielding an overall rate $f_{\mathrm{overall}} \triangleq p f_s$. To ensure the unique recovery of $\mathcal{X}$ from \eqref{quen:1}, i.e., no distinct $\mathcal{X}' \neq \mathcal{X}$ satisfies $\mathcal{A}\mathcal{X}' = \mathcal{A}\mathcal{X}$, the following uniqueness condition must hold~\cite{Chen2006}:
\begin{equation}\label{eq:rankkk}
	2|\operatorname{supp}(\mathcal{X})| < \operatorname{spark}(\mathcal{A}) - 1 + \operatorname{rank}(\mathcal{X}),
\end{equation}
where $\operatorname{spark}(\mathcal{A})$ denotes the smallest number of linearly dependent columns of $\mathcal{A}$. Since $\operatorname{rank}(\mathcal{X}) \geq 0$, a sufficient condition is $2|\operatorname{supp}(\mathcal{X})| < \operatorname{spark}(\mathcal{A}) - 1$. Noting that any matrix $\mathcal{A} \in \mathbb{C}^{p \times L}$ satisfies $\operatorname{spark}(\mathcal{A}) \leq p + 1$, the number of channels must satisfy
\begin{equation}\label{eq:9}
	p \geq \operatorname{spark}(\mathcal{A}) - 1 > 2|\operatorname{supp}(\mathcal{X})|,
\end{equation}
yielding the overall rate lower bound $f_{\mathrm{overall}} \geq 2|\operatorname{supp}(\mathcal{X})| f_s$. Let $\mathcal{B} \triangleq \sum_{j=1}^{N_{\mathrm{sig}}} B_{j}$ denote the total occupied bandwidth. In the typical sparse regime where $\mathcal{B} \ll f_{\mathrm{nyq}}$, the minimum rate for perfect recovery is $2 \mathcal{B}$~\cite{mishali2009blind}. Achieving this minimum requires
\begin{equation}\label{eq:12}
	|\operatorname{supp}(\mathcal{X})| f_s \leq \mathcal{B}.
\end{equation}
However, \eqref{eq:12} is generally unattainable in practice. When $f_s < B$, each subband spans multiple rows of $\mathcal{X}$ without fully occupying them, so the aligned $f_s$-grid overestimates the occupied bandwidth, resulting in $|\operatorname{supp}(\mathcal{X})| f_s$ on the order of $2\mathcal{B}$~\cite{song2022approaching, jang2018intentional}. Conversely, when $f_s > B$ (e.g., due to hardware limitations of the pseudorandom generator~\cite{Haque}), each active row carries at least $f_s - B$ of spectral redundancy, again forcing $|\operatorname{supp}(\mathcal{X})| f_s > \mathcal{B}$. In either case, a fundamental gap persists between the achievable rate and the theoretical minimum, motivating the design of improved sampling architectures.

\section{The proposed DAWC}

\subsection{The Spectral Analysis of DAWC}

To describe the sampling process of DAWC on $x(t)$, we adopt a frequency-domain perspective and trace how $X(f)$ evolves through each stage. The DAWC is depicted in Fig.~\ref{fig:3}, with it comprising $p+N_{\mathrm{sig}}$ parallel channels. In the first $p$ channels, the multiband signal $x(t)$ is modulated by a multitone signal $s_{\ell}(t)$, where $s_{\ell}(t)$ is defined as
\begin{equation}
	s_\ell(t) = \sum_{m=M_1 }^{ M_2 }\sum_{k=0 }^{ n-1 }a_{\ell,m,k} \,e^{-\mathrm{j}2\pi m f_p t-\mathrm{j}2\pi k f_c t}
\end{equation}
where $a_{\ell,m,k}\in \mathbb{C}$, $n \geq 2$ is a positive integer, $M_1$ and $M_2$ are given by $M_1 = -\left\lfloor \frac{L}{2} \right\rfloor$ and $M_2 = \left\lfloor \frac{L-1}{2} \right\rfloor $, respectively. Here, $L = \frac{f_{\mathrm{nyq}}}{f_p}$ is a positive integer representing the number of $f_p$-intervals that partition $\mathcal{F}_{\mathrm{nyq}}$. We require
$f_c < f_p$ so that each $f_p$-width interval can be further subdivided into finer segments. The FT of $s_\ell(t)$ satisfies
\begin{align}
	S_\ell(f) & = \int_{-\infty }^{+\infty} s_\ell(t)\,e^{-\mathrm{j}2\pi ft} dt \nonumber                                                           \\
	          & =\int_{-\infty }^{+\infty}\sum_{m=M_1 }^{ M_2 }\sum_{k=0 }^{ n-1 }a_{\ell,m,k}\,e^{\mathrm{j}2\pi( -f- m f_p - k f_c )t}dt \nonumber \\
	          & =\sum_{m=M_1 }^{ M_2 }\sum_{k=0 }^{ n-1 }a_{\ell,m,k}\,\delta(f+ m f_p+ k f_c)
\end{align}
where $\delta(\cdot)$ is the Dirac delta function. In the $\ell$-th channel, the mixed signal $\underline{x}_\ell(t) = x(t)s_\ell(t)$, and so its spectrum is given by the convolution of $X(f)$ and $S_\ell(f)$:
\begin{align}\label{eq:spectrum_sum}
	\underline{X}_\ell(f) & = X(f)*S_\ell(f) \nonumber                                             \\
	                      & =\int_{-\infty }^{+\infty} S_\ell(\tau-f)X(\tau)d\tau \nonumber        \\
	                      & = \sum_{m=M_1 }^{ M_2 }\sum_{k=0 }^{ n-1 }a_{\ell,m,k}X(f+m f_p+k f_c)
\end{align}
where the last line is due to the sifting property of $\delta(\cdot)$ that $\int X(\tau) \delta(\tau - f) d\tau = X(f)$. From~\eqref{eq:spectrum_sum}, the modulated spectrum $\underline{X}_\ell(f)$ is a superposition of $nL$ frequency-shifted copies of $X(f)$, each shifted by $mf_p+kf_c$. In the filtering stage, $\underline{X}_\ell(f)$ is passed through an LPF with cut-off frequency of $\mathcal{W}_{\mathrm{LPF}}$, yielding the filtered output:
\begin{align}
	Y_\ell(f) & =   \underline{X}_\ell(f)H_\ell(f) \nonumber                                                                    \\
	          & =   \sum_{m=M_1 }^{ M_2 }\sum_{k=0 }^{ n-1 }a_{\ell,m,k}X(f+m f_p+k f_c)H_\ell(f) \vspace{3mm} \label{eq:eq:23}
\end{align}
where
\begin{equation}
	H_\ell(f) =
	\begin{cases}
		1, & f \in [0, \mathcal{W}_{\mathrm{LPF}}], \\
		0, & \text{otherwise}.
	\end{cases}
\end{equation}

\begin{figure}[t]
	\centerline{\includegraphics[scale = 0.387]{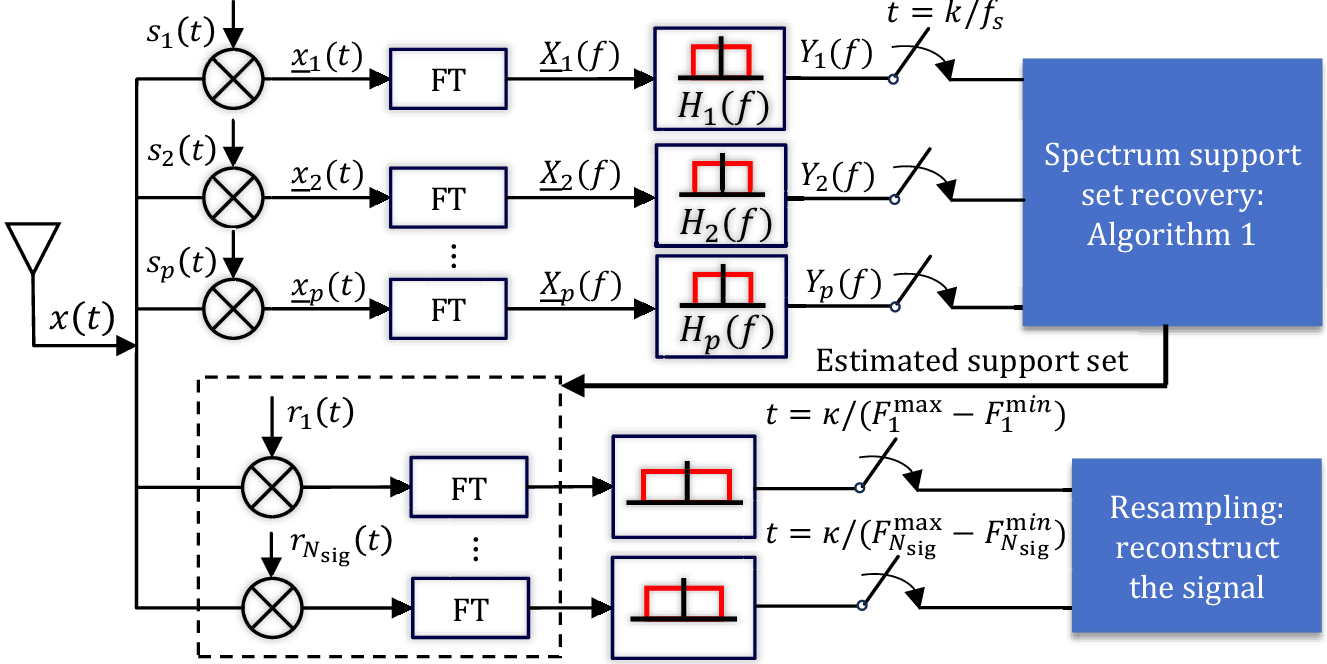}}
	\vspace{-1mm}
	\caption{The DAWC utilizes $p$ channels at rate $f_s$ to identify the spectrum support set, and $N_{\mathrm{sig}}$ channels at rate $B+2f_c$ to reconstruct the signal.}
	\label{fig:3} \vspace{-3mm}
\end{figure}

Subsequently, $Y_\ell(f)$ is digitized by an ADC operating at a sampling rate of $f_s = \mathcal{W}_{\mathrm{LPF}}$. Compared with~\eqref{eq:spectrum_sum}, the only difference in \eqref{eq:eq:23} is that the frequency-shifted components are now confined to $[0,\mathcal{W}_{\mathrm{LPF}}]$. The term $a_{\ell,m,k}X(f+m f_p+k f_c)H_\ell(f)$ in~\eqref{eq:eq:23} can be expressed as $a_{\ell,m,k}X(f)$ for $f \in [-m f_p-k f_c, \mathcal{W}_{\mathrm{LPF}}-mf_p-k f_c]$. This implies that each component in~\eqref{eq:eq:23} corresponds to a segment of $X(f)$ of width $\mathcal{W}_{\mathrm{LPF}}$. To maximize the information from $X(f)$ preserved in $Y_\ell(f)$, it is desirable that the frequency-shifted components occupy disjoint intervals, thereby avoiding redundant coverage of the same spectral content. In the following lemma, we provide the required conditions on $f_s$, $f_c$, $f_p$ and $n$ that ensure the components occupy disjoint intervals.
\begin{lemma}\label{lemma_1}
	For $f_s<f_c<f_p$, when $f_s+(n-1)f_c \leq f_p$, the components of $Y_\ell(f)$ in~\eqref{eq:eq:23} are pairwise disjoint.
\end{lemma}
\emph{Proof}: See Appendix~\ref{appendix_a}.

Under the conditions of Lemma~\ref{lemma_1}, the output of the first $p$ channels is a weighted sum of non-overlapping spectrum components of $X(f)$. Stacking the outputs $Y_\ell(f)$ from the $p$ channels row-wise yields the matrix-form representation:
\begin{equation}
	\mathbf{y}(f) = [ Y^\top_{1}(f),  Y^\top_{2}(f),  \ldots,  Y^\top_{p}(f) ]^{\top}
\end{equation}
where $\mathbf{y}(f) \in \mathbb{C}^{p}$ is a vector-valued function of $f$. Since the spectral components in~\eqref{eq:eq:23} share a common structure across channels, differing only in the coefficients $a_{\ell,m,k}$, they can be expressed in the following unified matrix form:
\begin{align}\label{eq:xfff}
	\mathbf{x}(f) = \Big[ & X\big(f - \lfloor\tfrac{L}{2}\rfloor f_p \big)^\top,
	X\big(f - \lfloor\tfrac{L}{2}\rfloor f_p + f_c\big)^\top, \nonumber                                        \\
	                      & \ldots, X\big(f + \lfloor\tfrac{L-1}{2}\rfloor f_p + (n-1)f_c\big)^\top \Big]^\top
\end{align}
where $f\in [0,f_s]$ and $\mathbf{x}(f) \in \mathbb{C}^{nL}$. The components in $\mathbf{x}(f)$ are arranged in lexicographic order over $(m, k)$ with $m \in \{M_1, \dots, M_2\}$ and $k \in \{0, \dots, n-1\}$, where $k$ serves as the fast-varying index.
\begin{figure*}[b]
	\hrule
	\begin{eqnarray}\label{eq:30}
		\small{ \underbrace{\left[\begin{array}{cccc}\hspace{-2mm}Y_{1}\left(0 \right) &\hspace{-2mm}Y_{1}\left(\frac{f_s}{rN-1} \right) & \hspace{-2mm} \cdots \hspace{-2mm} & Y_{1}\left(f_s  \right) \hspace{-2mm} \\ \hspace{-2mm} Y_{2}\left(0  \right) &\hspace{-2mm}Y_{2}\left(\frac{f_s}{rN-1} \right) & \hspace{-2mm}\cdots \hspace{-2mm}&  Y_{2}\left(f_s  \right) \hspace{-2mm}  \\ \hspace{-2mm} \vdots   &\hspace{-2mm} \ddots \hspace{-2mm}    &  \vdots \hspace{-2mm} \\ \hspace{-2mm} Y_{p}\left(0  \right) &\hspace{-2mm}Y_{p}\left(\frac{f_s}{rN-1} \right) &  \hspace{-2mm} \cdots \hspace{-2mm} & Y_{p}\left(f_s   \right) \hspace{-2mm} \end{array}\right ]}_{\triangleq \mathbf{Y} \in \mathbb{C}^{p \times Nr}} \hspace{-0.8mm} = \hspace{-0.8mm}\underbrace{\left[ \begin{array}{ccc}\hspace{-1.7mm}a_{1, -\left\lfloor\frac{L}{2} \right\rfloor,0 }  \hspace{-1.5mm} &\hspace{-1.7mm} \cdots \hspace{-1.7mm} & a_{1, \left\lfloor \frac{L-1}{2} \right\rfloor,n-1} \hspace{-1.7mm}\\ \hspace{-1.7mm}a_{2, -\left\lfloor\frac{L}{2} \right\rfloor,0 }\hspace{-1.5mm} & \hspace{-1.7mm}\cdots\hspace{-1.7mm} & a_{2, \left\lfloor \frac{L-1}{2} \right\rfloor,n-1}  \hspace{-1.7mm} \\ \hspace{-1.7mm}\vdots \hspace{-1.5mm}   & \hspace{-1.7mm} \ddots \hspace{-1.7mm} & \vdots  \hspace{-1.7mm}\\ \hspace{-1.7mm}a_{p, -\left\lfloor\frac{L}{2} \right\rfloor,0 } \hspace{-1.5mm} &\hspace{-1.7mm} \cdots\hspace{-1.7mm} & a_{p, \left\lfloor \frac{L-1}{2} \right\rfloor,n-1} \hspace{-1.7mm}\end{array}\right]}_{\triangleq \mathbf{A} \in \mathbb{C}^{p\times nL}} \hspace{-0mm}\underbrace{\left [ \begin{array}{ccc}\hspace{-2mm}X \left( -\lfloor\frac{L}{2}\rfloor f_p  \right) & \hspace{-2.5mm} \cdots \hspace{-2.5mm} & X\left(-\lfloor\frac{L}{2}\rfloor f_p \hspace{-1.2mm}+\hspace{-1.2mm}f_s   \right) \hspace{-2.5mm} \\ \hspace{-2mm}  \vdots & \hspace{-2.5mm} \ddots \hspace{-2.5mm} &\vdots \hspace{-2.5mm} \\ \hspace{-2mm} X\left((n\hspace{-0.4mm}-\hspace{-0.4mm}1)f_c \hspace{-0.4mm}-\hspace{-0.4mm} \lfloor\frac{L}{2}\rfloor f_p \right)& \hspace{-2.5mm}\cdots \hspace{-2.5mm} & X\left((n\hspace{-0.4mm}-\hspace{-0.4mm}1)f_c \hspace{-0.4mm}-\hspace{-0.4mm} \lfloor\frac{L}{2}\rfloor f_p\hspace{-0.4mm} + \hspace{-0.4mm} f_s  \right) \hspace{-2.5mm} \\ \hspace{-2mm} \vdots &\hspace{-2.5mm} \ddots \hspace{-2.5mm} &\vdots \hspace{-2.5mm} \\ \hspace{-2mm} X\left((n \hspace{-0.4mm}-\hspace{-0.4mm}1)f_c \hspace{-0.4mm}+\hspace{-0.4mm}\lfloor\frac{L-1}{2}\rfloor f_p \right)& \hspace{-2.5mm} \cdots \hspace{-2.5mm} & X\left((n\hspace{-0.4mm}-\hspace{-0.4mm}1)f_c \hspace{-0.4mm}+\hspace{-0.4mm} \lfloor\frac{L-1}{2}\rfloor f_p \hspace{-0.4mm} + \hspace{-0.4mm}f_s   \right) \hspace{-2mm}\end{array}\right]}_{\triangleq \mathbf{X} \in \mathbb{C}^{nL \times Nr}}
		}
	\end{eqnarray}
\end{figure*}

Collecting the coefficients from all $p$ channels into a single matrix, we obtain the sensing matrix $\mathbf{A} \in \mathbb{C}^{p \times nL}$ such that $\mathbf{y}(f) = \mathbf{A}\,\mathbf{x}(f)$. The $(i,j)$-th entry of $\mathbf{A}$ is given by $\mathbf{A}_{i,j} = a_{i,\left \lfloor \frac{j}{n}  \right \rfloor,j-n\left \lfloor \frac{j}{n} \right \rfloor}$. Assuming that the filtered output is observed over a time window of length $T$, sampling at rate $f_s$ yields a total of $Tf_s$ samples. We further assume that $Tf_s$ can be factored as $Tf_s = rN$, where $N$ is the number of samples per snapshot and $r$ is the number of snapshots. This leads to the discretized input-output relationship in~\eqref{eq:30}.

Fig.~\ref{Fig4} shows an example of the input-output relationship of the DAWC. The spectrum partitioning strategy of DAWC is illustrated on the left panel of Fig.~\ref{Fig4}. Within each interval of width $f_p$, $n$ spectrum segments of width $f_s$ are selected. These selected segments are then weighted and aggregated. The right panel of Fig.~\ref{Fig4} demonstrates how specific elements from the CSSS signal matrix $\mathcal{X}$ are selected and reconfigured to constitute the matrix $\mathbf{X}$. Specifically, the $n$ column submatrices in $\mathcal{X}$ are rearranged row-wise to form the matrix $\mathbf{X}$.

\begin{figure*}[h]
	\centering
	\includegraphics[scale = 0.19]{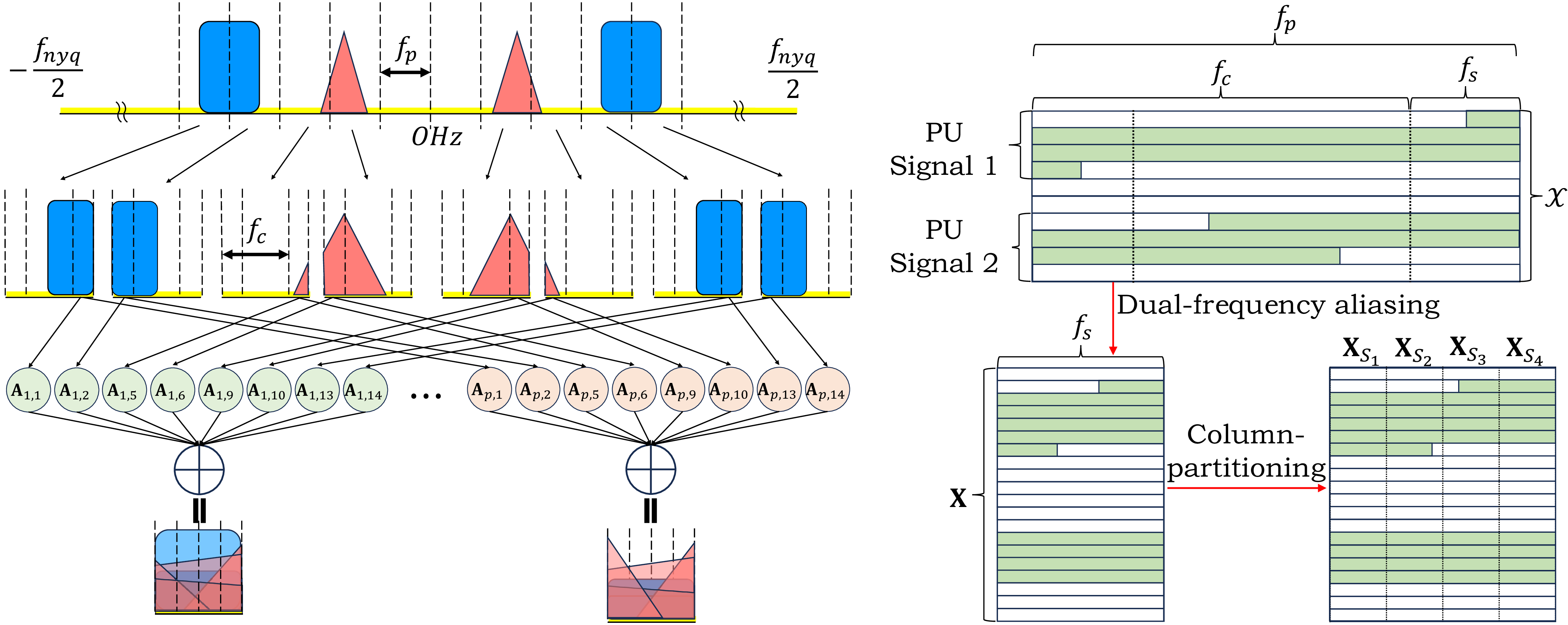}
	\caption{Illustration of the DAWC and the column-partitioned model in~\eqref{eq:sub-MMVproblem}. Left: spectral evolution from the input spectrum $X(f)$ through modulation, filtering, and sampling to the channel outputs, which are then uniformly partitioned into column blocks. Right: relationship between the signal matrices of the CSSS and DAWC, depicting the column-partitioned structure defined in~\eqref{eq:sub-MMVproblem}.} \label{Fig4}
\end{figure*}

\subsection{The Signal Reconstruction of DAWC}
In the CSSS, the matrix $\mathcal{X}$ comprises all $L$ spectral components of $X(f)$, so recovering $\mathcal{X}$ from~\eqref{quen:1} is equivalent to reconstructing $X(f)$. This recovery proceeds in two stages: first, the support set $\operatorname{supp}(\mathcal{X})$ is identified; then, the signal matrix is reconstructed via least squares as $\mathcal{X}^* = \mathbf{A}^{\dagger}_{S}\mathcal{Y}$, where $S$ denotes an estimate of $\operatorname{supp}(\mathcal{X})$.

In the DAWC, the signal matrix $\mathbf{X}$ defined in~\eqref{eq:30} represents the discrete spectrum samples of~\eqref{eq:xfff}. We show that, under certain conditions, the support $\operatorname{supp}(\mathbf{X})$ can be used to infer the approximate locations of all subbands.
\begin{lemma}\label{lemma_2}
	For any $\xi \in [nL]$, the function $\mathcal{F}(\xi)$ identifies the starting frequency associated with the $\xi$-th row of $\mathbf{X}$:
	\begin{equation}
		\mathcal{F}(\xi) = \Big( \big\lceil \tfrac{\xi}{n} \big\rceil - \big\lfloor \tfrac{L}{2} \big\rfloor - 1 \Big)f_p + \big( (\xi-1) \bmod n \big) f_c.
		\label{eq:freq_map}
	\end{equation}
	Each subband $[f^{\min}_{j}, f^{\max}_{j})$ corresponds to a block of consecutive indices $\{\alpha_j, \alpha_j+1, \dots, \beta_j\}$ of the matrix $\mathbf{X}$. For $f_s<f_c<f_p$, when $f_s+(n-1)f_c = f_p$ and $f_c-f_s < \min(B_j)$ for $j \in [N_{\mathrm{sig}}]$, we have the following inequalities:
	\begin{equation}\label{lemm21}
		\mathcal{F}(\alpha_j)-f_c+f_s< f^{\min}_{j} \leq \mathcal{F}(\alpha_j)+f_s,
	\end{equation}
	\begin{equation}\label{lemm22}
		\mathcal{F}(\beta_j) \leq f^{\max}_{j} < \mathcal{F}(\beta_j)+f_c.
	\end{equation}
\end{lemma}

\emph{Proof}: See Appendix~\ref{appendix_b}.

Lemma~\ref{lemma_2} shows that the bounds of each subband can be inferred from $\operatorname{supp}(\mathbf{X})$. Specifically, under the conditions of Lemma~\ref{lemma_2}, the $j$-th subband satisfies $[f^{\min}_j, f^{\max}_j) \subseteq [F^{\min}_j, F^{\max}_j)$, where $F^{\min}_j \triangleq \mathcal{F}(\alpha_j) - f_c + f_s$ and $F^{\max}_j \triangleq \mathcal{F}(\beta_j) + f_c$ are the estimated lower and upper bounds, respectively. Given these bounds for all $j \in [N_{\mathrm{sig}}]$, each subband is shifted to baseband, low-pass filtered, and sampled to reconstruct its spectrum. As shown in Fig.~\ref{fig:3}, this is accomplished through the last $N_{\mathrm{sig}}$ dedicated reconstruction channels, where the modulating waveform of the $j$-th channel is given by
\begin{equation}\label{alisaing_signal}
	r_j(t)=e^{-\mathrm{j}2\pi \big(\frac{F^{\max}_j+F^{\min}_j}{2}\big)t}.
\end{equation}

The modulating waveform in~\eqref{alisaing_signal} shifts the subband center frequency $\frac{F^{\max}_j+F^{\min}_j}{2}$ to baseband\footnote{When two subbands (e.g., the $j$-th and $(j+1)$-th) are closely spaced, the ranges $[F^{\min}_j,F^{\max}_j]$ and $[F^{\min}_{j+1},F^{\max}_{j+1}]$ may overlap. In such cases, it suffices to take their union before applying the shift in~\eqref{alisaing_signal}.}. The shifted signal is then passed through a bandpass filter with passband $\big[\frac{F^{\min}_j-F^{\max}_j}{2},\, \frac{F^{\max}_j-F^{\min}_j}{2}\big]$ and sampled at rate $F^{\max}_j - F^{\min}_j$, yielding $\mathcal{M}_j \triangleq T(F^{\max}_j - F^{\min}_j)$ spectral samples denoted by $\hat{\mathbf{y}}_j \in \mathbb{C}^{\mathcal{M}_j}$. The frequency resolution of the $j$-th channel is $\Delta f_j = \frac{F^{\max}_j - F^{\min}_j}{\mathcal{M}_j - 1} $ for $j \in [N_{\mathrm{sig}}]$, and the $k$-th sample $\hat{\mathbf{y}}_{j,k}$ for $k \in \{0, \dots, \mathcal{M}_j - 1\}$ corresponds to the physical frequency $\frac{F^{\max}_j+F^{\min}_j}{2} + k\,\Delta f_j$. Via spectral synthesis, the reconstructed signal $\hat{x}(t)$ for $t \in [0,\,T]$ is given by
\begin{equation}
	\hat{x}(t)
	= \sum_{j=1}^{N_{\mathrm{sig}}}\frac{1}{\mathcal{M}_j}
	\sum_{k=0}^{\mathcal{M}_j-1}
	\hat{\mathbf{y}}_{j,k}\,
	e^{
			\mathrm{j}\,2\pi
			\big(\frac{F^{\min}_j-F^{\max}_j}{2} + k\,\Delta f_j\big)\, t
		}. \nonumber
\end{equation}

\subsection{Reconstruction Error Analysis}

Let $\hat{\mathcal{T}} = \bigcup_{j=1}^{\hat{K}} [a_j, b_j]$ denote the estimated spectral support, where $\hat{K}$ is the estimated number of subbands. The signal is reconstructed over $\hat{\mathcal{T}}$ and set to zero elsewhere. Define the correctly covered, missed, and falsely sampled regions as $\mathcal{D} = \hat{\mathcal{T}} \cap \mathcal{T}$, $\mathcal{U} = \mathcal{T} \setminus \hat{\mathcal{T}}$, and $\mathcal{V} = \hat{\mathcal{T}} \setminus \mathcal{T}$, respectively, where $\mu(\cdot)$ denotes the Lebesgue measure (i.e., total length) of a union of intervals.

The reconstruction error satisfies $\hat{X}(f) - X(f) = W(f)$ for $f \in \hat{\mathcal{T}}$ and $\hat{X}(f) - X(f) = -X(f)$ for $f \in \mathcal{U}$, where $\hat{X}(f)$ is the reconstructed spectrum and $W(f)$ is the reconstruction noise. Assuming a uniform noise PSD $S_W$ across the sampled region and an approximately flat signal PSD $\bar{P}_x$ within each subband, the integrated MSE is
\begin{equation}\label{eq:mse_cont}
	\mathrm{MSE} = S_W \cdot \mu(\hat{\mathcal{T}}) + \bar{P}_x \cdot \mu(\mathcal{U}).
\end{equation}
Normalizing by the total signal energy $\bar{P}_x \cdot \mathcal{B}$ and defining the bandwidth detection ratio $\rho_d = \mu(\mathcal{D}) / \mathcal{B}$, the bandwidth false alarm ratio $\rho_f = \mu(\mathcal{V}) / (W - \mathcal{B})$, and the per-frequency SNR $\mathrm{SNR}_f = \bar{P}_x / S_W$, we obtain
\begin{equation}\label{eq:nmse}
	\mathrm{NMSE} = \tfrac{1}{\mathrm{SNR}_f}\left(\rho_d + \tfrac{W - \mathcal{B}}{\mathcal{B}}\,\rho_f\right) + (1 - \rho_d).
\end{equation}

\begin{remark}\label{remark:nmse}
	From~\eqref{eq:nmse}, the NMSE is governed by two terms: a noise-induced error scaled by $1/\mathrm{SNR}_f$ and the missed bandwidth fraction $1 - \rho_d$. At high SNR, the latter dominates, so reconstruction quality depends primarily on how well $\hat{\mathcal{T}}$ covers $\mathcal{T}$, while false alarms incur only a negligible penalty of order $\rho_f/\mathrm{SNR}_f$. With perfect support estimation ($\rho_d = 1$, $\rho_f = 0$), the NMSE reduces to the noise floor $1/\mathrm{SNR}_f$.
\end{remark}

\subsection{The Sampling Rate of DAWC}
The overall sampling rate of DAWC consists of two parts: the rate $pf_s$ for recovering $\operatorname{supp}(\mathbf{X})$, and the rate $\sum_{j=1}^{N_{\mathrm{sig}}}(F^{\max}_j - F^{\min}_j)$ for reconstructing the active subbands. We now characterize bounds on each part. Subtracting the left-hand side of~\eqref{lemm21} from the right-hand side of~\eqref{lemm22} gives $B_j < \mathcal{F}(\beta_j) - \mathcal{F}(\alpha_j) + 2f_c - f_s$. Similarly, subtracting the right-hand side of~\eqref{lemm21} from the left-hand side of~\eqref{lemm22} yields $\mathcal{F}(\beta_j) - \mathcal{F}(\alpha_j) - f_s \leq B_j$. Combining these two inequalities with $F^{\max}_j - F^{\min}_j = \mathcal{F}(\beta_j) - \mathcal{F}(\alpha_j) + 2f_c - f_s$, we obtain
\begin{equation}
	B_j \leq F^{\max}_j - F^{\min}_j < B_j + 2f_c.
\end{equation}

Summing over all $N_{\mathrm{sig}}$ subbands, the reconstruction rate is upper bounded by $\mathcal{B}+2N_{\mathrm{sig}}f_c$, where $\mathcal{B} = \sum_{j=1}^{N_{\mathrm{sig}}} B_j$. Meanwhile, as established in~\eqref{eq:9}, unique recovery of $\operatorname{supp}(\mathbf{X})$ requires a minimum rate of $2|\operatorname{supp}(\mathbf{X})|f_s$. Consequently, when the number of channels $p$ is chosen at the minimum required for unique recovery, the overall sampling rate is upper bounded by $2|\operatorname{supp}(\mathbf{X})|f_s + \mathcal{B}+2N_{\mathrm{sig}}f_c$. As discussed in Section~\ref{Section_II_C}, the overall rate of the CSSS inevitably exceeds $2\mathcal{B}$. The following theorem establishes the conditions under which the upper bound attains the minimum.

\begin{theorem}\label{theorem_1}
	When $ 2\sum_{j=1}^{N_{\mathrm{sig}}} \bigg(n\left \lfloor \frac{B_j}{f_p} \right \rfloor + \left \lceil \frac{ \operatorname{mod}(B_j, f_p)  }{f_c} \right  \rceil+1 \bigg )f_s +2N_{\mathrm{sig}}f_c \leq \mathcal{B}$ and $f_c = \frac{f_p - f_s}{n-1}$, for $f_s<f_c<f_p$, we have
	\begin{equation}
		2|\operatorname{supp}(\mathbf{X})|f_s + 2N_{\mathrm{sig}}f_c \leq \mathcal{B}.
	\end{equation}
\end{theorem}
\emph{Proof}: See Appendix~\ref{appendix_c}.

\begin{remark}\label{remark:rate}
	For each subband, applying $\lfloor x \rfloor \leq x$ and $\lceil x \rceil \leq x + 1$ yields
	\begin{equation}
		n\left\lfloor \tfrac{B_j}{f_p} \right\rfloor + \left\lceil \tfrac{\operatorname{mod}(B_j,\, f_p)}{f_c} \right\rceil + 1 \leq \tfrac{n}{f_p}\, B_j + \tfrac{1}{f_c}\operatorname{mod}(B_j,\, f_p) + 2. \nonumber
	\end{equation}
	Since $\operatorname{mod}(B_j, f_p) \leq B_j$ and $\max\{n/f_p,\, 1/f_c\}$ dominates both coefficients, summing over all $N_{\mathrm{sig}}$ subbands gives
	\begin{equation}
		\sum_{j=1}^{N_{\mathrm{sig}}} \bigg(n\left\lfloor \tfrac{B_j}{f_p} \right\rfloor + \left\lceil \tfrac{\operatorname{mod}(B_j,\, f_p)}{f_c} \right\rceil + 1 \bigg) \leq \mathcal{B} \cdot \max\left\{ \tfrac{n}{f_p},\, \tfrac{1}{f_c} \right\}. \nonumber
	\end{equation}
	Combining this with Theorem~\ref{theorem_1}, the parameters $f_s$, $f_p$, $f_c$, and $n$ can be chosen such that the overall sampling rate attains $2\mathcal{B}$ for a given multiband signal.
\end{remark}
\section{The MSSP Algorithm}\label{section_4}
In practice, environmental and device noise are accounted for by adding a noise term to~\eqref{eq:30}:
\begin{equation}\label{eq:37}
	\mathbf{Y}=\mathbf{A}\mathbf{X}+\mathbf{E},
\end{equation}
where $\mathbf{E} \in \mathbb{C}^{p \times rN}$ denotes the noise matrix. Under the assumption that $\mathbf{X}$ is $s$-sparse (i.e., $|\operatorname{supp}(\mathbf{X})| = s$), this section addresses the robust recovery of $\mathbf{X}$ from the noisy measurements. To illustrate how MSSP identifies the support set $\operatorname{supp}(\mathbf{X})$, the formulation in~\eqref{eq:37} is partitioned column-wise into equal intervals, yielding
\begin{equation}\label{eq:sub-MMVproblem}
	\mathbf{Y} _{ S_{i} }=\mathbf{A}  \mathbf{X}_{ S_{i} } + \mathbf{E}_{S_{i}},\ \ i \in [r],
\end{equation}
where
$S_i \triangleq \left\{ (i-1)N+1, \ldots, iN \right\}$, $\mathbf{Y}_{S_{i}} \in \mathbb{C}^{p \times N}$,  $\mathbf{X}_{S_{i}} \in \mathbb{C}^{L \times N}$ and $\mathbf{E}_{S_{i}}\in \mathbb{C}^{p \times N}$ are the $i$-th column-partitioned sub-matrices of $\mathbf{Y}$, $\mathbf{X}$ and $\mathbf{E}$, respectively. Fig.~\ref{Fig4} illustrates the decomposition of $\mathbf{X}$ into $r$ uniform submatrices, whose support sets exhibit substantial overlap. This implies that if the support set of one submatrix is known, the remaining submatrices are expected to share a similar sparsity pattern. We treat this shared sparsity structure as side information (SI) and exploit it in our algorithm. For notational simplicity, we omit the subscript $S_i$ in this section and write the concatenated form directly. For instance, $\mathbf{X} = [\mathbf{X}_{S_1}, \ldots, \mathbf{X}_{S_r}]$ and $\mathbf{X}^{k} = [\mathbf{X}^{k}_{S_1}, \ldots, \mathbf{X}^{k}_{S_r}]$ (see Step~8 of Alg.~\ref{alg:MSSP}).

\subsection{Definitions and Lemmas}
Prior to introducing the MSSP, we define the RIP and provide some useful lemmas, with proofs in the appendices.
\begin{definition} \label{Definition_1}
	A sensing matrix $\mathbf{A} \in \mathbb{C}^{p\times L} $ is said to satisfy the $s$-order RIP if for any $s$-sparse $(\| \mathbf{x} \|_0 \le s )$ vector $\mathbf{x} \in \mathbb{C}^{L}$
	\begin{equation}\label{eq:RIP}
		\left ( 1-\delta  \right ) \left \| \mathbf{x} \right \|_{2}^{2}  \le \left \| \mathbf{A} \mathbf{x} \right \|_{2}^{2}  \le \left ( 1+\delta  \right ) \left \| \mathbf{x} \right \|_{2}^{2},
	\end{equation}
	where $0 \le \delta < 1$. The infimum of $\delta$, denoted by $\delta_{s}$, is called the restricted isometry constant (RIC) of $\mathbf{A}$.
\end{definition}

\begin{lemma} \label{lemma_3}
	For nonnegative numbers $a$, $b$, $c$, $d$, $x$, $y$,
	\begin{equation}\label{eq:nonnegative}
		a(b x+c y)^2 + dy^2 \leq (\sqrt{a}bx +\sqrt{ac^2+d}y)^2.
	\end{equation}
\end{lemma}

\begin{lemma}\label{lemma_4}
	Consider the concatenated matrix $\mathbf{X}^{(\tilde{S}^{k})^c} = [\mathbf{X}_{S_1}^{(\tilde{S}^{k}_{S_1})^c}, \dots, \mathbf{X}_{S_r}^{(\tilde{S}^{k}_{S_r})^c}]$. In Steps 4 and 5 of the MSSP algorithm, for any $ \omega \in [0,1]$, the following relationship holds:
	\begin{equation}\label{eq:32}
		\| \mathbf{X}^{(\tilde{S}^{k})^c  }  \|_{F}  \leq \nu _{1} \left \| \mathbf{X}-\mathbf{X}^{k-1} \right \|_{F} + \sqrt{2(1+\delta _{3s})}\left \| \mathbf{E} \right \| _{F}
	\end{equation}
	where $\nu _{1} \triangleq 1-\omega+\omega\delta_{3s}$.
\end{lemma}

\begin{lemma}\label{lemma_5}
	Consider the model~\eqref{eq:37} with $ |\operatorname{supp}(\mathbf{X})| = s$. Let $T \subseteq [nL]$ be a set with $\left | T \right | =t\geq 2s$ and $T_{0} \subseteq [nL]$ denote the SI index set. Define $T_{t-s}$ as the set of the smallest $t-s$ entries of $ \sum_{j \in T \cap T_0 } \omega' \| \mathbf{P}^{\perp }_{T\setminus \{j\}}\mathbf{Y} \| ^{2}_{F} + \sum_{j \in T\setminus T_0 }  \| \mathbf{P}^{\perp }_{T\setminus \{j\}}\mathbf{Y} \| ^{2}_{F}$ for $j \in T$ and $\omega'=1+\omega$ where $\omega \in [0,1]$. If the matrix $\mathbf{A}$ satisfies the RIP with $\delta_t <1$, we have
	\begin{equation}\label{eq:19}
		\| \mathbf{X}^{T_{t-s}}  \|^{2}_{F}\leq \nu _{2}\sqrt{1+\delta_s}\| \mathbf{X}^{T^c } \| ^{2}_{F}+   \nu _{2} \| \mathbf{E} \|^{2}_{F}
	\end{equation}
	where $\nu_2 \triangleq \tfrac{\sqrt{\omega'(1+\delta_t)}+\sqrt{1-\delta_t}}{1-\delta_t}$.
\end{lemma}

\begin{lemma}\label{lemma_6}
	Let $\mathbf{W}_{T_0} \in \mathbb{C}^{L \times L}$ be a diagonal matrix whose $i$-th diagonal entry is $1$ if $i \in T_0$ and $\omega \in [0,1]$ if $i \notin T_0$, where $T_0 \subseteq [L]$. For a matrix $\mathbf{V} \in \mathbb{C}^{L \times N}$, let $V \triangleq \operatorname{supp}(\mathbf{V})$. Suppose $U \subseteq [L]$ with $|U| \leq s$. If $|V| \leq t$ and $\mathbf{A}$ satisfies the RIP with $\delta_{t+s} < 1$,
	\begin{equation}\label{eq:merge_imp_eq3}
		\|  ((\mathbf{I}_{L}-\mathbf{W}_{T_{0}}\mathbf{A}^{H }\mathbf{A} )\mathbf{V})^{ U}  \|_{F} \le (1-\omega+\delta_{t+s})\| \mathbf{V}  \|_F.
	\end{equation}
\end{lemma}

\begin{lemma}\label{lemma_7}
	For the model~\eqref{eq:37}, let $T_{0}\subseteq [nL]$, $U\subseteq [nL] $ and $\left | T_{0} \cap U \right | \le u$, we have
	\begin{equation}\label{eq:lemma7}
		\| (\mathbf{W}_{T_{0}}\mathbf{A}^{H } \mathbf{E})^{ U}  \| _{F}\le \sqrt{1+ \delta _{u}} \left \| \mathbf{E} \right \| _{F}.
	\end{equation}
\end{lemma}

\subsection{Observation and Computational Complexity Analysis}
We begin with an observation on the measurement matrix in the inputs of MSSP. As shown in Alg.~\ref{alg:MSSP}, the measurement matrix $\mathbf{Y}$ is partitioned into $r$ column submatrices of the same size. At the $k$-th iteration, in Step 3, the SI set for estimating $\operatorname{supp}(\mathbf{X}_{S_{i}})$ is chosen as the union of the $r-1$ previously estimated supports. Put formally,
\begin{equation}
	\Lambda_{S_{i}}^k = {\bigcup}_{j \in [r] \setminus i } S^{k - 1}_{S_{j}}.
\end{equation}

Except for Step 3, the recovery processes of the $r$ submatrices are mutually independent. Initialized with a residual matrix $\mathbf{R}^{0}_{S_{i}} = \mathbf{Y}_{S_{i}}$ and an estimated support $S^0_{S_{i}} = \emptyset$ for $i \in [r]$, it iteratively constructs the support sets of $\mathbf{X}_{ S_{i} }$ for $i \in [r]$ (and also the estimated signal submatrices) through a merging-and-pruning strategy (see Step 4 and Step 6). This strategy is mainly inspired by~\cite{dai2009subspace }.

\begin{algorithm}[t] \label{alg:MSSP}
	\caption{The MSSP Algorithm} \small
	\SetKwInOut{Input}{Input}
	\SetKwInOut{Output}{Output}
	\SetKwInOut{Initialize}{Initialize}

	\Input{sensing matrix $\mathbf{A} \in \mathbb{C}^{p \times nL}$ with unit-column vector, measurements $\mathbf{Y}_{S_{i}} \in \mathbb{C}^{p \times N}$, $i \in [r]$, sparsity $s$, weighting parameter $\omega \in [0,1] $ and residual tolerance $\epsilon$. }

	\Initialize{iteration count $k = 0$, residual matrix $\mathbf{R}^{0}_{S_{i}} = \mathbf{Y}_{S_{i}}$ and estimated support $S^0_{S_{i}} = \emptyset$, $i \in [r]$. }

	\While{$\| \mathbf{R}^{k}_{S_{i}} \| _{F} < \epsilon$}
	{$k = k + 1$;

	{\bf Compute} SI set $\Lambda_{S_{i}}^k = \bigcup_{j \in [r] \setminus i } S^{k - 1}_{S_{j}}$;

	{\bf Select} $\Delta S = \big \{ s$ \rm{indices corresponding to the largest $s$ values among all elements in} $\| \mathbf{A}_{:,j}\mathbf{A}^{\dagger}_{:,j} \mathbf{R}^{k - 1}_{S_{i}} \|_{F}^{2}$ \rm{for} $j\in \Lambda_{S_{i}}^k$ and $\| \omega\mathbf{A}_{:,j}\mathbf{A}^{\dagger}_{:,j}\mathbf{R}^{k - 1}_{S_{i}} \|_{F}^{2}$ \rm{for} $j\in (\Lambda_{S_{i}}^k)^c$ $\big \}$;

	{\bf Merge} $\tilde{S}^{k}_{S_{i}} = S^{k - 1}_{S_{i}} \cup \Delta S$;



	{\bf Prune} $S^{k}_{S_{i}} = \big \{ s$ \rm{indices corresponding to the largest $s$ values of} $ \sum_{j \in \tilde{S}^{k}_{S_{i}} \cap \Lambda_{S_{i}}^k } \omega' \| \mathbf{P}^{\perp }_{\tilde{S}^{k}_{S_{i}}\setminus \{j\}}\mathbf{Y}_{S_{i}} \| ^{2}_{F} + \sum_{j \in \tilde{S}^{k}_{S_{i}}\setminus \Lambda_{S_{i}}^k }  \| \mathbf{P}^{\perp }_{\tilde{S}^{k}_{S_{i}}\setminus \{j\}}\mathbf{Y}_{S_{i}} \| ^{2}_{F} \big \}$ where $\omega'=1+\omega$;

	{\bf Estimate} $\mathbf{X}^{k}_{S_{i}} = \underset {\boldsymbol{\Theta}: \operatorname{supp}(\boldsymbol{\Theta}) = {S}^{k}_{S_{i}} } {\arg \min} \left\|\mathbf{Y}_{S_{i}}-\mathbf{A}  \boldsymbol{\Theta}  \right\|_{F} $;

	{\bf Update} residual matrix $\mathbf{R}^{k}_{S_{i}} = \mathbf{Y}_{S_{i}} - \mathbf{A} \mathbf{X}^{k}_{S_{i}}$;
	}
	\Output{estimated signal submatrices $\mathbf{X}^{k}_{S_{i}}$'s.}
\end{algorithm}


To properly incorporate the SI, a weighting projection
\begin{equation}\label{eq:select}
	\begin{cases}
		\|\mathbf{A}_{:,j}\mathbf{A}^{\dagger}_{:,j} \mathbf{R}^{k - 1}_{S_{i}} \|_{F}^{2}, & j\in \Lambda_{S_{i}}^k     \\ \|\omega\mathbf{A}_{:,j}\mathbf{A}^{\dagger}_{:,j}\mathbf{R}^{k - 1}_{S_{i}} \|_{F}^{2},
		                                                                                    & j\in (\Lambda_{S_{i}}^k)^c
	\end{cases}
\end{equation}
is constructed with a weighting parameter $ \omega \in [0,1]$ indexed by $\Lambda_{S_{i}}^k$ being 1, and $\omega \geq 0$ otherwise in Step 4. The weighting projection of~\eqref{eq:select} implies that the supports in the estimation sets of the other $r-1$ sub-matrices are more likely to be selected. The set $\Delta S$ chosen in Step 4 will be merged with the estimation set $S^{k-1}_{S_i}$ from the previous iteration. To exclude erroneous elements from the union set $\tilde{S}_{S_i}^{k}$, we introduce a SI-pruning operation that reduces the probability of excluding elements in the SI set. In Step 6, we project the measurement values onto the subspaces formed by the incomplete candidate support $\tilde{S}^{k}_{S_{i}}\setminus \{j\}$. Specifically, we construct
\begin{equation}\label{eq:punrrrr}
	\sum_{j \in \tilde{S}^{k}_{S_{i}} \cap \Lambda_{S_{i}}^k } \omega' \| \mathbf{P}^{\perp }_{\tilde{S}^{k}_{S_{i}}\setminus \{j\}}\mathbf{Y}_{S_{i}} \| ^{2}_{F} \hspace{-0.5mm} + \sum_{j \in \tilde{S}^{k}_{S_{i}}\setminus \Lambda_{S_{i}}^k }  \| \mathbf{P}^{\perp }_{\tilde{S}^{k}_{S_{i}}\setminus \{j\}}\mathbf{Y}_{S_{i}} \| ^{2}_{F}
\end{equation}
for pruning $s$ indices from $\tilde{S}^{k}_{S_{i}}$ where $\omega'=1+\omega$. The estimation operation in Step 7 involves minimizing the representation distance of the Frobenius norm to $\mathbf{Y}_{S_{i}}$. The estimated matrix $\mathbf{X}^{k}_{S_{i}}$ is supported on the pruned set $S^{k}_{S_{i}}$. The algorithm terminates until the residual matrices meet the tolerance $\epsilon \geq 0$, and outputs the estimated submatrices\footnote{Since there are $r$ submatrices, the termination condition can be set such that the Frobenius norms of all submatrices are less than a certain threshold.}.

\begin{remark}[Computational complexity]\label{remark:complexity}
	The per-iteration cost of MSSP is dominated by Step~4 with $\mathcal{O}(p^2 nNL)$ operations per submatrix, the merge step with $\mathcal{O}(s)$, and Step~6 with $\mathcal{O}\big(p^2(s^2 + Ns)\big)$. Empirically, the algorithm converges in $\mathcal{O}(s)$ iterations, yielding a total complexity of $\mathcal{O}(srp^2 nNL)$ across $r$ submatrices.
\end{remark}

\subsection{Main Results of the MSSP}
We now analyze the recovery performance of MSSP under the RIP framework (Definition~\ref{Definition_1}).

\begin{figure}[t]
	\centerline{\includegraphics[width=0.98\linewidth]{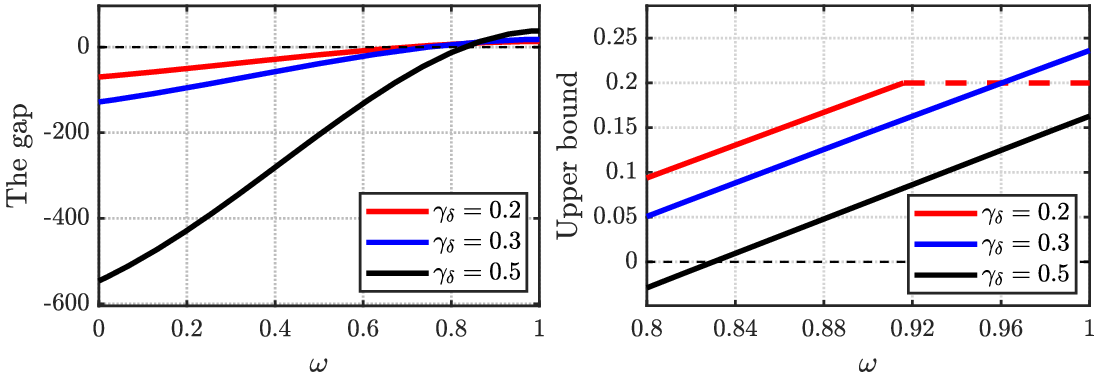}}
	\vspace{-3mm}
	\caption{Numerical analysis of the gap in inequality \eqref{eq:4444YTT} and the derived upper bound of $\delta_{3s}$ in~\eqref{eq:omga_to_meet}.}
	\label{fig:5}
\end{figure}

\begin{theorem}\label{theorem_3}
	Consider the column-partitioned MMV model~\eqref{eq:sub-MMVproblem} with $| \operatorname{supp}(\mathbf{X}_{S_{i}})|\leq s$. For a constant $\gamma_{\delta}\in(0,1)$, there exist a weighting parameter $\omega \in[0, 1]$ satisfies
	\begin{equation}\label{eq:4444YTT}
		1 + C(\gamma_{\delta},\omega)\omega(2 - \omega)>C^2(\gamma_{\delta},\omega)(1 - \omega)^2
	\end{equation}
	where $C(\gamma_{\delta},\omega) \triangleq 1+ \frac{(\sqrt{(1+\omega)(1+\gamma_{\delta})}+\sqrt{1-\gamma_{\delta}})^2(1+\gamma_{\delta})}{(1-\gamma_{\delta})^2}$. \begin{figure*}[b]
		\hrule\vspace{1mm}
		\begin{equation}\label{eq:omga_to_meet}
			\delta_{3s}  < \min \bigg\{\gamma_{\delta},\frac{ \sqrt{1 + C(\gamma_{\delta},\omega)\omega(2 - \omega)}-C(\gamma_{\delta},\omega)(1 - \omega)}{C(\gamma_{\delta},\omega) + 1} \bigg\}
		\end{equation}
	\end{figure*} If the $3s$-RIC $\delta_{3s}$ of $\mathbf{A}$ satisfies~\eqref{eq:omga_to_meet}, the MSSP algorithm produces an estimate $\mathbf{X}^{k}=[\mathbf{X}^{k}_{S_{1}},\cdots,\mathbf{X}^{k}_{S_{r}}]$ satisfying
	\begin{equation}
		\left \| \mathbf{X}-\mathbf{X}^{k} \right \|_{F} \le \rho ^{k}\left \| \mathbf{X} \right \|_{F} +\tau \left \| \mathbf{E} \right \|_{F},
	\end{equation}
	where $\rho \in (0, 1)$ and $\tau$ are constants defined in~\eqref{eq:rho} and~\eqref{eq:tau}, respectively. 
	Furthermore, after at most
	$
		k^{ \ast }= \left\lceil \log_{ \rho }{\frac{\tau \| \mathbf{E}   \|_{F} }{ \| \mathbf{X} \|_{F}} } \right\rceil
	$
	iterations, 
	\begin{equation} \label{eq:stable}
		\| \mathbf{X}-\mathbf{X}^{k^{ \ast }}   \|_{F} \le (\tau + 1) \| \mathbf{E}   \|_{F}.
	\end{equation}
\end{theorem}


From Theorem~\ref{theorem_3}, it can be seen that when
$\|\mathbf{E}\|_F = 0$ (noiseless case), the estimation
error $\|\mathbf{X} - \mathbf{X}^{k}\|_{F}$ converges to $0$
as $k \to \infty$. Our main interest lies in the value of $\omega$ for which~\eqref{eq:4444YTT} holds, as well as the upper bound of $\delta_{3s}$, given the specified $\gamma_{\delta}$. The left panel of Fig.~\ref{fig:5} depicts the gap between the left-hand side and the right-hand side of~\eqref{eq:4444YTT} with respect to $\omega \in [0,1]$. Clearly, when $\omega$ exceeds a certain threshold,~\eqref{eq:4444YTT} holds true. In the right panel of Fig.~\ref{fig:5}, we illustrate the upper bound of $\delta_{3s}$ given by~\eqref{eq:omga_to_meet} for the parameters satisfying~\eqref{eq:4444YTT}. It can be observed that the upper bound is monotonically non-decreasing with respect to $\omega$. The proof of Theorem~\ref{theorem_3} proceeds by establishing a recursive relationship between the estimation error at the current iteration and that of the previous iteration.

	{\it Proof of Theorem~\ref{theorem_3}:} First, we consider one iteration and bound $\| \mathbf{X}^{(\tilde{S}^{k})^c  }  \|_{F} $ with $\|\mathbf{X} \hspace{-.25mm} - \hspace{-.25mm} \mathbf{X}^{k-1} \|_F$. By Lemma~\ref{lemma_4}, we have
\begin{equation}\label{eq:47}
	\| \mathbf{X}^{(\tilde{S}^{k})^c  }  \|_{F} \le \nu _{1} \left \| \mathbf{X}-\mathbf{X}^{k-1} \right \|_{F}  + \sqrt{2(1+\delta _{3s})}\left \| \mathbf{E} \right \| _{F}.
\end{equation}
where $\nu _{1} =1-\omega+\delta_{3s}$. In Step 6, for the $i$-th submatrix, let $S^{t-s}_{S_i}\triangleq\tilde{S}^k_{S_i} \setminus S^k_{S_i}$ denote the index set corresponding to the smallest $s$ entries. For $i\in [r]$, taking $T=\tilde{S}^k_{S_i}$ and $T_{t-s}=S_{S_i}^{t-s}$, it follows from Lemma~\ref{lemma_5} that
\begin{equation}\label{eq:50}
	\| \mathbf{X}^{ S^{t-s}_{S_i}}_{S_i}  \|_{F}\leq \nu_2\sqrt{1+\delta_s}\| \mathbf{X}_{S_i}^{(\tilde{S}^k)^c } \| _{F} +   \nu_2\| \mathbf{E}_{S_i} \|_{F}.
\end{equation}
where $\nu_2 \triangleq \tfrac{\sqrt{\omega'(1+\delta_{2s})}+\sqrt{1-\delta_{2s}}}{1-\delta_{2s}}$. Recalling the definition of $S^{t-s}_{S_i}$ and the fact that $\|\mathbf{X}^{A \cup B}\|^{2}_{F} = \|\mathbf{X}^{A}\|^{2}_{F} + \|\mathbf{X}^{B}\|^{2}_{F}$ for any disjoint sets $A \cap B = \emptyset$, obtain
\begin{align}
	 & \|\mathbf{X}^{({S^{k}})^c}\|_F^2 =\|\mathbf{X}^{S^{t-s}}\|_F^2+ \|\mathbf{X}^{({\tilde{S}^{k}})^c} \|_F^2  \nonumber                                                                       \\
	 & \stackrel{\eqref{eq:50}}{\leq}\nu _{2}^2(\sqrt{1+\delta_s} \|\mathbf{X}^{({\tilde{S}^{k}})^c} \|_F + \left\|\mathbf{E}\right\|_F)^2 + \|\mathbf{X}^{({\tilde{S}^{k}})^c} \|_F^2  \nonumber \\
	 & \stackrel{~\eqref{eq:nonnegative}}{\le} \big(\sqrt{1+\nu_2^2(1+\delta_s)}\|\mathbf{X}^{({\tilde{S}^{k}})^c} \|_F + \nu_2\|\mathbf{E}\|_F\big)^2 \nonumber                                  \\
	 & \stackrel{\eqref{eq:47}}{\leq}\big(\sqrt{1+\nu_2^2(1+\delta_s)} (\nu _{1} \left \| \mathbf{X}-\mathbf{X}^{k-1} \right \|_{F} \nonumber                                                     \\
	 & \quad \quad\quad+ \sqrt{2\varrho^2}\left \| \mathbf{E} \right \| _{F}) + \nu_2\|\mathbf{E}\|_F\big)^2
\end{align}
where $ \varrho \triangleq \sqrt{1+\delta_{3s}} $. Using the monotonicity $\delta_s \leq \delta_{3s}$
\cite{CIte47} and squaring both sides, we obtain
\begin{align}\label{eq:S.29}
	\|\mathbf{X}^{({S^{k}})^c}\|_F & \leq   \sqrt{\nu _{1}^2+\nu _{1}^2\nu_2^2\varrho^2}\left \| \mathbf{X}-\mathbf{X}^{k-1} \right \|_{F} \nonumber \\  &\quad + \big(\sqrt{2\varrho^2+2\nu_2^2\varrho^4} + \nu_2 \big)\|\mathbf{E}\|_F
\end{align}

Step~8 addresses a least squares problem. By setting $S = S^k$, $\tilde{\mathbf{X}} = \mathbf{X}^k$, and $t = s < 2s$, and invoking the monotonicity $\delta_{2s} \leq \delta_{3s}$, it follows from \cite[Lemma~5]{Bing} that
\begin{equation}\label{eq:S.30}
	\|\mathbf{X}-\mathbf{X}^k \|_F \leq \sqrt{\frac{1}{1-\delta_{3s}^2}} \|\mathbf{X}^{(S^{k})^c } \|_F+\frac{\sqrt{1+\delta_{3s}}}{1-2\delta_{3s}} \|\mathbf{E} \|_F.
\end{equation}

Since $\nu_2$ is monotonically increasing in $\delta_{2s}$, we have $\nu_2 \leq \nu_3 \triangleq \tfrac{\sqrt{\omega'(1+\delta_{3s})}+\sqrt{1-\delta_{3s}}}{1-\delta_{3s}}$. Combining \eqref{eq:S.29} and \eqref{eq:S.30} yields
\begin{equation}\label{eq:oneiteration}
	\| \mathbf{X}-\mathbf{X}^{k}   \|_{F} \leq \rho \| \mathbf{X}-\mathbf{X}^{k-1}   \|_{F} + (1-\rho) \tau\| \mathbf{E} \|_{F}
\end{equation}
where
\begin{equation}\label{eq:rho}
	\rho \triangleq \sqrt{\frac{\nu_1^2+\nu_1^2\nu_3^2\varrho^2}{1-\delta _{3s}^{2}}},
\end{equation}
\begin{equation}\label{eq:tau}
	(1-\rho)\tau \triangleq \frac{\sqrt{2\varrho^2+2\nu_2^2\varrho^4}+\nu_2}{\sqrt{1-\delta_{3s}^2}}+\frac{\sqrt{1+\delta_{3s}}}{1-2\delta_{3s}}.
\end{equation}

When $\rho < 1$, by recursively applying \eqref{eq:oneiteration}, we obtain
\begin{equation}
	\| \mathbf{X}-\mathbf{X}^{k}   \|_{F} \le \rho ^{k}   \| \mathbf{X}   \|_{F} + (1-\rho^k)\tau   \| \mathbf{E}\|_{F}.
\end{equation}
Furthermore, noting that the coefficient $(1-\rho^k)\tau$ is upper-bounded by $\tau$, we have
\begin{equation}
	\| \mathbf{X}-\mathbf{X}^{k}   \|_{F} \le \rho ^{k}   \| \mathbf{X}   \|_{F} +\tau   \| \mathbf{E}\|_{F}.
\end{equation}
If $\rho<1$, by~\eqref{eq:rho}, the $3s$-RIC $\delta_{3s}$ obeys $\nu_1^2+\nu_1^2\nu_3^2\varrho^2+\delta_{3s}^2-1<0$. After arrangement, one finds
\begin{equation}\label{eq:inequality}
	\delta_{3s}^2+C(\delta_{3s},\omega)(1-\omega+\delta_{3s})^2-1<0
\end{equation}
where $C(\delta_{3s},\omega) = 1+ \frac{(\sqrt{(1+\omega)(1+\delta_{3s})}+\sqrt{1-\delta_{3s}})^2(1+\delta_{3s})}{(1-\delta_{3s})^2}$ is a function of $\delta_{3s}$ and $\omega$. Let $\delta_{3s}< \gamma_{\delta} < 1$, since $C(\delta_{3s},\omega)$ is monotonically increasing with respect to $\delta_{3s}$, a sufficient condition for \eqref{eq:inequality} is obtained by substituting $\gamma_{\delta}$ for $\delta_{3s}$ in $C(\delta_{3s},\omega)$, yielding
\begin{equation}\label{eq:inequality2}
	\delta_{3s}^2+C(\gamma_{\delta},\omega)(1-\omega+\delta_{3s})^2-1<0.
\end{equation}
Solving this inequality shows that it holds when
$\delta_{3s}$ satisfies
\begin{eqnarray}\label{eq:conmdidi}
	\delta_{3s} < \frac{ \sqrt{1 + C(\gamma_{\delta},\omega)\omega(2 - \omega)}-C(\gamma_{\delta},\omega)(1 - \omega)}{C(\gamma_{\delta},\omega) + 1}.
\end{eqnarray}

The condition $\delta_{3s} < \gamma_{\delta}$ ensures that~\eqref{eq:omga_to_meet} is satisfied. Since $0 \leq \delta_{3s} < \gamma_{\delta} < 1$, a necessary and sufficient condition for the existence of such a $\delta_{3s}$ satisfying~\eqref{eq:conmdidi} reduces to~\eqref{eq:4444YTT}. It remains to verify that there exists an $\omega \in [0,1]$ satisfying~\eqref{eq:4444YTT} for any $\gamma_{\delta} \in (0,1)$. By the monotonicity of $C(\gamma_{\delta}, \omega)$, we have $C(\gamma_{\delta}, \omega) > 1$ for all $\gamma_{\delta} \in (0,1)$ and $\omega \in [0,1]$. For any fixed $\gamma_{\delta}$, as $\omega$ increases from $0$ to $1$, the left-hand side of~\eqref{eq:4444YTT} increases from $1$ to a value exceeding $1$, while the right-hand side decreases from a value greater than $1$ at $\omega = 0$ to $0$ at $\omega = 1$. Since both sides are continuous in $\omega$, the intermediate value theorem guarantees the existence of an $\omega^* \in [0,1]$ satisfying~\eqref{eq:4444YTT}. Regarding the number of iterations required for MSSP, there exists a $k^{ \ast }= \left\lceil \log_{ \rho }{\frac{\tau \| \mathbf{E}   \|_{F} }{ \| \mathbf{X} \|_{F} } } \right\rceil$ such that $\rho ^{k}   \| \mathbf{X}   \|_{F} \leq \tau   \| \mathbf{E}   \|_{F}$. This completes the proof.

\section{Numerical Experiments}\label{num_exper}
\subsection{Parameter Setting}
In this section, we conduct numerical experiments to evaluate the performance of the proposed DAWC and MSSP. We consider a multiband signal $x(t)$ spanning $[-5, 5]$~GHz, comprising $N_{\mathrm{sig}} \in \{1, \dots, 7\}$ subband signals modeled as a superposition of modulated sinc pulses:
\begin{equation}\label{eq:sincsignal}
	x(t)=\sum_{i=1}^{N_{\mathrm{sig}}}A_i\,\mathrm{sinc}(B_i(t-t_i))\,e^{\mathrm{j}2\pi f_i t},
\end{equation}
where $A_i$, $t_i$, and $f_i$ denote the amplitude, time delay, and carrier frequency of the $i$-th subband, respectively, and $B_i \in \{50, 100, 150,200\}$~MHz is the corresponding bandwidth. The carrier frequencies are drawn uniformly at random from $[-5000, 5000]$~MHz, and the time delays are also randomly selected. White Gaussian noise is added at SNR levels of $\{0, 10, 20\}$~dB. As an illustrative example, we set $N_{\mathrm{sig}} = 3$ with $f_i \in \{-2300, 1700, 2500\}$~MHz. Figs.~\ref{fig:6} and~\ref{fig:7} depict the amplitude spectrum of the original noise-free signal and the received signal at an SNR of $20$~dB, respectively.

\begin{figure}[t]
	\centerline{\includegraphics[width=0.9\linewidth]{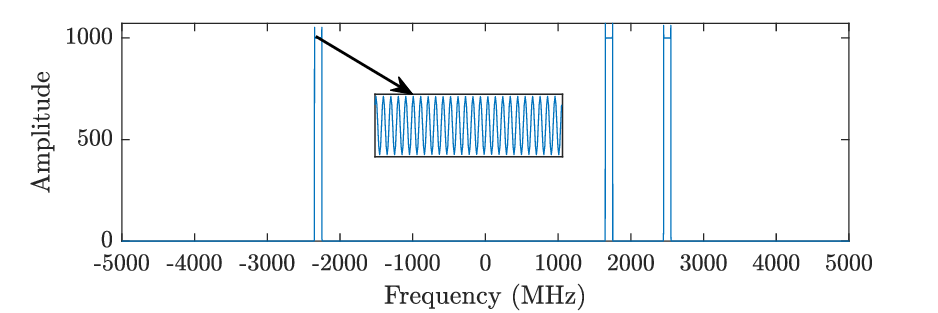}}
	\vspace{-3mm}
	\caption{The amplitude of the original noise-free signal spectrum.}
	\label{fig:6}
\end{figure}
\begin{figure}[t]\vspace{-4.5mm}
	\centerline{\includegraphics[width=0.9\linewidth]{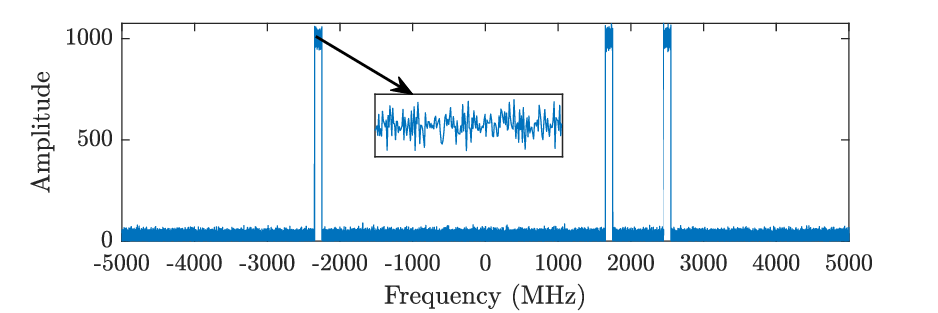}}
	\vspace{-3mm}
	\caption{The amplitude of the noisy signal spectrum as SNR$=20$dB.}
	\label{fig:7}
\end{figure}

We compare the proposed DAWC with MWC and MCS in terms of the sampling architecture, and with four sparse recovery algorithms in terms of reconstruction performance: SOMP~\cite{tropp2006algorithms}, MP (extended to joint sparsity)~\cite{Matching}, SSMP~\cite{Jian}, and subspace pursuit (extended to joint sparsity)~\cite{Bing}. Both MWC and DAWC employ random Gaussian sensing matrices, while MCS uses a deterministic one. Support recovery is evaluated via the detection probability $P_d \triangleq \frac{|T \cap S|}{|T|}$ and false alarm probability $P_f \triangleq \frac{|S \setminus T|}{L - |T|}$, where $T$ and $S$ denote the true and estimated support sets, respectively. Reconstruction performance is evaluated under varying parameters $\omega$ and sparsity levels $s$, with algorithm-specific stopping criteria\footnote{SOMP and MP execute a fixed number of $s$ iterations, while the remaining methods terminate once the residual energy stabilizes.}. All results are averaged over 2000 independent Monte Carlo trials.

\subsection{Performance Gain Delivered by DAWC Scheme}

\begin{figure*}[h]
	\centering
	\includegraphics[scale = 0.29]{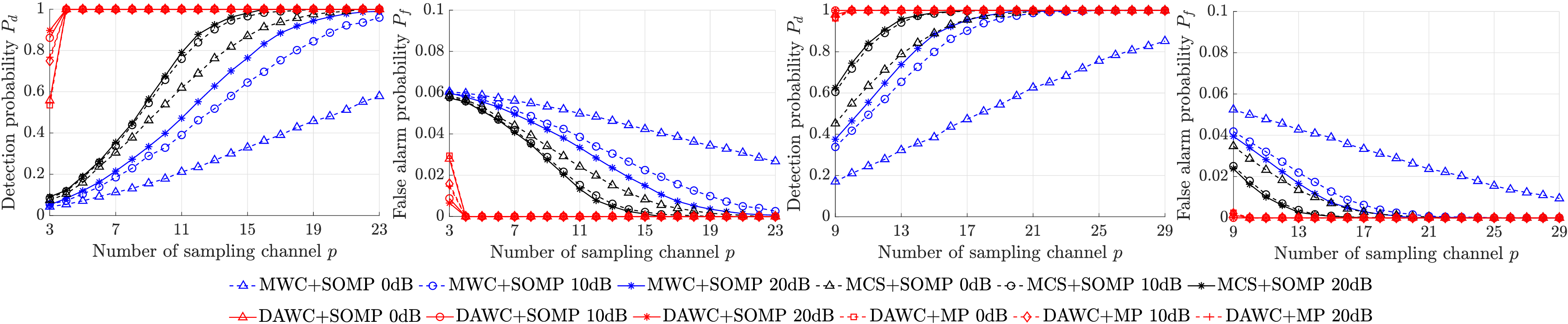}
	\vspace{-5pt}
	\caption{Comparison of $P_d$ and $P_f$ for MWC, MCS, and DAWC under varying SNRs. The first two panels: subband bandwidth $B_i=50$~MHz with $p \in [3, 23]$. The last two panels: subband bandwidth $B_i\in \{100,150,200\}$~MHz with $p \in [9, 29]$.}  \label{Fig6}
\end{figure*}

\begin{figure*}[h]
	\centering
	\vspace{-3mm}
	\includegraphics[scale = 0.29]{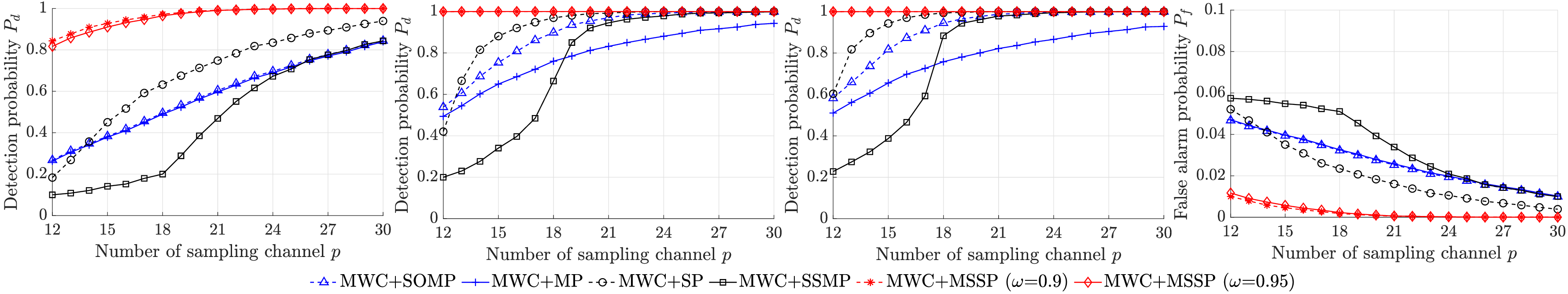}
	\vspace{-5pt}
	\caption{Performance comparison ($P_d$ and $P_f$) of the proposed MSSP and benchmarks (SOMP, MP, SP, SSMP) versus the number of channels $p \in [12, 30]$ of the MWC. The first three panels depict $P_d$ at $\text{SNR} \in \{0, 10, 20\}$~dB, respectively, while the last panel shows $P_f$ at $\text{SNR} = 0$~dB.}   \label{Fig7}
\end{figure*}

\begin{figure*}[t]
	\vspace{-5pt}
	\centering
	\begin{minipage}[t]{0.48\textwidth}
		\centering
		\includegraphics[width=\textwidth]{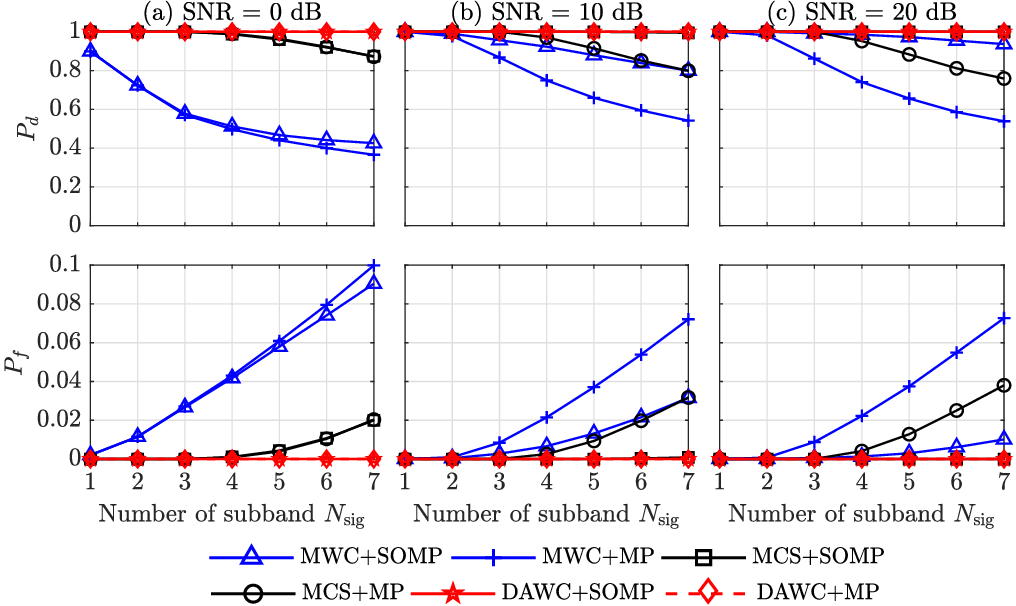}
		\caption{Performance comparison ($P_d$ and $P_f$) of MWC, MCS and DAWC architectures under different numbers of subbands.}
		\label{fig:mwc_aawc}
	\end{minipage}
	\hfill
	\begin{minipage}[t]{0.48\textwidth}
		\centering
		\includegraphics[width=\textwidth]{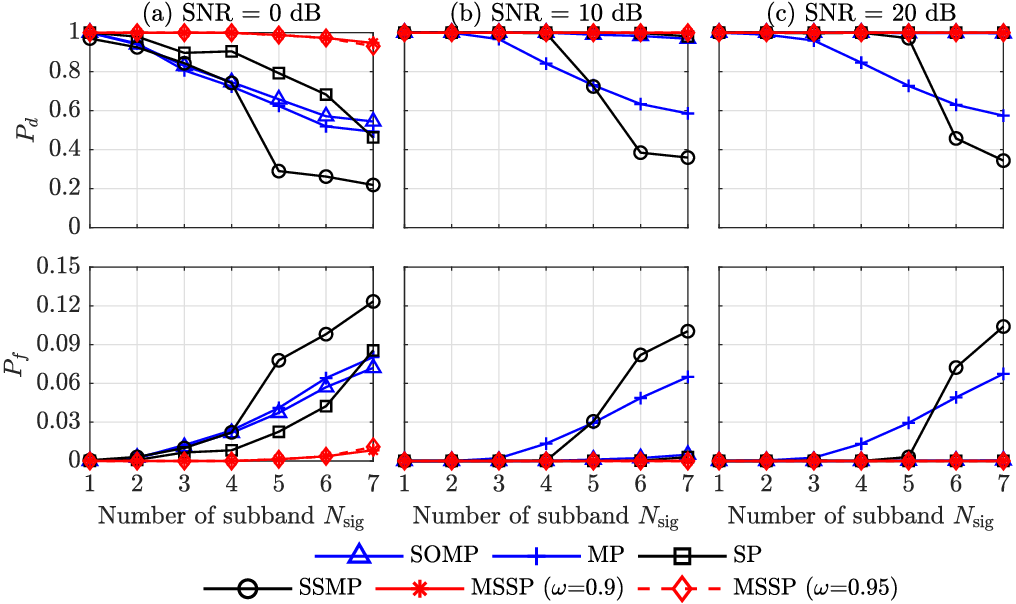}
		\caption{Performance comparison ($P_d$ and $P_f$) of six algorithms under different numbers of subbands.}
		\label{fig:6alg}
	\end{minipage}
	\vspace{-10pt}
\end{figure*}

In this subsection, we demonstrate the performance superiority of the proposed DAWC, particularly under low sampling rates. Two bandwidth scenarios are considered: (i)~$B_i = 50$~MHz for all subbands, and (ii)~$B_i \in \{100, 150, 200\}$~MHz, with DAWC parameters $(f_c, f_s, f_p, n) = (19.8~\text{MHz}, 1~\text{MHz}, 100~\text{MHz}, 6)$ and $(32~\text{MHz}, 4~\text{MHz}, 100~\text{MHz}, 4)$, respectively. The overall sampling rate ranges from $300$ to $2300$~MHz for scenario~(i) and from $900$ to $2900$~MHz for scenario~(ii). For convenience, we express the sampling rate in units of ``sampling channels'', where one channel corresponds to $100$~MHz.

As illustrated in Fig.~\ref{Fig6}, the proposed DAWC exhibits superior robustness compared to the benchmarks. Most notably, at the theoretical minimum sampling rate (i.e., $f_{\mathrm{overall}} = 2\sum_{i}B_i = 300$~MHz for $B_i = 50$~MHz and $f_{\mathrm{overall}} = 900$~MHz for $B_i \in \{100, 150, 200\}$~MHz), DAWC achieves high $P_d$ and low $P_f$, whereas the competing methods suffer from severe performance degradation. This performance gain is attributed to the targeted sampling strategy of DAWC within specific spectrum intervals of width $f_p$, which effectively mitigates sampling redundancy by focusing solely on informative bands. In contrast, the baseline methods adopt an indiscriminate full-spectrum acquisition approach, inevitably capturing a substantial amount of vacant spectrum.


\subsection{Performance Gain Delivered by MSSP Algorithm}
In this subsection, we evaluate the recovery performance of different algorithms implemented within the MWC framework. The subband bandwidths are set to $B_i \in \{50, 100, 150\}$~MHz. For the MSSP, the parameters are configured as $r=6$ and $\omega \in \{0.9, 0.95\}$. Fig.~\ref{Fig7} depicts the performance comparison against benchmarks under varying sampling rates and SNRs.

As illustrated in Fig.~\ref{Fig7}, MSSP exhibits substantial superiority over the comparative algorithms, particularly in the low sampling rate regime where it achieves a significantly higher probability of support recovery. Across varying SNRs, MSSP maintains a clear performance lead. Specifically, at $0$~dB, it outperforms other algorithms despite the generally limited reconstruction accuracy across all methods. At $10$~dB and $20$~dB, it achieves near-perfect support recovery even at low sampling rates. This performance gain is attributed to the fact that MSSP effectively exploits the inherent sparse correlation of the signal matrix and employs a robust support discrimination mechanism. Regarding the weight parameter $\omega$, the performance gap is marginal, with $\omega=0.9$ yielding slightly better results than $\omega=0.95$.

\subsection{Impact of the Number of Subbands}
In time-varying scenarios, the number of active subbands may change rapidly over time. This subsection evaluates the proposed DAWC and MSSP against benchmark methods under varying subband counts at fixed overall sampling rates. Specifically, Fig.~\ref{fig:mwc_aawc} considers a subband bandwidth of $50$~MHz with an overall sampling rate of $2300$~MHz, while Fig.~\ref{fig:6alg} considers $100$~MHz with $3000$~MHz.

As shown in Figs.~\ref{fig:mwc_aawc} and~\ref{fig:6alg}, both proposed schemes exhibit strong reconstruction robustness under dynamic subband counts. The proposed architecture achieves perfect reconstruction across all tested SNR levels, while the proposed algorithm only fails at $0$~dB SNR with $7$ subbands. In all cases, both schemes significantly outperform the baselines.

\section{CONCLUSION}
In this paper, we have proposed a complete sub-Nyquist sampling framework for wideband spectrum sensing, comprising DAWC and MSSP. The proposed DAWC selectively samples the informative spectral intervals of the wideband spectrum while discarding non-essential spectral content during the filtering stage, thereby enabling low-rate acquisition. By leveraging the accurately estimated subband locations, perfect waveform reconstruction is achieved at a sampling rate only marginally above the total occupied bandwidth. Complementing this, we have developed two precise support identification strategies for the signal matrix, forming the basis of MSSP. Moreover, by incorporating SI into the iterative process, MSSP achieves robust support recovery under noisy conditions. Simulation results substantiate the effectiveness and superiority of the proposed scheme.

In future work, we will extend the proposed framework to multiband signals sparse in the fractional Fourier domain, investigate sampling strategies for signals with time-varying spectral support, and explore incorporating SI into deep unfolding-based algorithms to further improve performance.

\appendices
\section{Proof of Lemma~\ref{lemma_1}}\label{appendix_a}
\begin{proof} For the LPF with a passband of $[0,f_s)$, only the components of $Y_{\ell}(f)$ in this range are non-zero. Replace $H_{\ell}(f)$ by $1$ in $[0,f_s)$, the corresponding parts of $X(f)$ follow
	\begin{align}\label{app_a_1}
		 & X(f+mf_p+kf_c)H_{\ell}(f),f\in [0,f_s] \nonumber \\
		 & = X(f),f\in [mf_p+kf_c,mf_p+kf_c+f_s]
	\end{align}
	where $m\in \{M_1,\ldots,M_2\}$ and $k \in \{0,\ldots,n-1\}$. We first prove the partial components in~\eqref{app_a_1} are pairwise disjoint for a fixed $m$ (let $m=R_0$). The components are given by
	\begin{equation}\label{app_a_2}
		\sum_{k=0 }^{ n-1 }X(f),f\in [R_0f_p+kf_c,R_0f_p+kf_c+f_s].
	\end{equation}

	It suffices to ensure that the adjacent components indexed by $k=Z_0$ and $k=Z_0+1$ are disjoint for all $Z_0 \in \{0, \dots, n-2\}$. This implies that the upper bound of the interval $k=Z_0$ is less than or equal to the lower bound of $k=Z_0+1$, yielding:
	\begin{equation} \label{app_a_3}
		R_0 f_p+Z_0 f_c+f_s \leq R_0 f_p+(Z_0+1) f_c.
	\end{equation}

	Under \eqref{app_a_3} ($f_s \leq f_c$), the intervals in \eqref{app_a_2} are inherently disjoint for a fixed $R_0$. To extend this property to distinct $R_0$, we consider the upper bound of the frequency range. Given $k \in \{0, \dots, n-1\}$, the maximum frequency for a given $R_0$ is $R_0 f_p + (n-1)f_c + f_s$. Ensuring non-overlapping bands between $R_0$ and $R_0+1$ requires this maximum to not exceed the lower bound corresponding to $m=R_0+1$. We have
	\begin{equation} \label{app_a_4}
		R_0 f_p+(n-1) f_c+f_s \leq (R_0+1) f_p.
	\end{equation}
	From~\eqref{app_a_4}, obtain $(n-1) f_c+f_s \leq f_p$, which proves Lemma~\ref{lemma_1}.
\end{proof}

\section{Proof of Lemma~\ref{lemma_2}}\label{appendix_b}
\begin{proof}
	First, we prove that each subband must intersect with the intervals corresponding to a block of consecutive indices of $\mathbf{X}$. According to~\eqref{eq:30}, the intervals are given by
	\begin{equation}
		[mf_p+kf_c,mf_p+kf_c+f_s]
	\end{equation}
	for $m\in \{M_1,\ldots,M_2 \}$ and $k\in \{0,\ldots,n-1 \}$. Under the condition $f_s + (n-1)f_c = f_p$, Lemma~\ref{lemma_1} implies the existence of $n$ pairwise disjoint intervals within $[mf_p, (m+1)f_p]$, as shown in Fig.~\ref{fig:aa}. For any subband $[ f_{j}^{\min}, f_{j}^{\max} )$, we identify indices $R_1$ and $Z_1$ such that $R_1f_p + Z_1f_c \leq f_{j}^{\min} < R_1f_p + (Z_1+1)f_c$. Similarly, there exist $R_2$ and $Z_2$ such that $f_s$ satisfies $R_2f_p + Z_2f_c \leq f_{j}^{\max} < R_2f_p + (Z_2+1)f_c$.
	\begin{figure}[t]
		\centerline{\includegraphics[scale = 0.45]{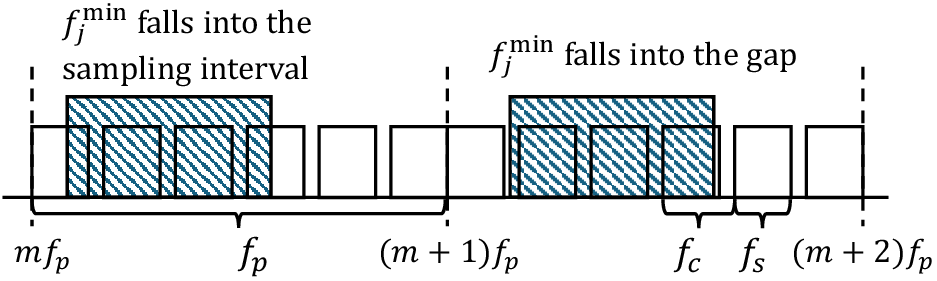}}
		\vspace{-1mm}
		\caption{Schematic illustration of the relationship between sampling intervals and subbands. Blank intervals ($f_s$) represent sampling, while filled intervals denote subbands.}
		\label{fig:aa} \vspace{-3mm}
	\end{figure}
	If $f_{j}^{\min} \leq R_1f_p + Z_1f_c + f_s$, the subband intersects with at least one frequency interval. Conversely, if $f_{j}^{\min}$ falls in the gap (i.e., $> R_1f_p + Z_1f_c + f_s$), we deduce that $f_{j}^{\max} = f_{j}^{\min} + B_j \overset{(a)}{>} f_{j}^{\min} + f_c - f_s \overset{(b)}{>} R_1f_p + (Z_1+1)f_c$\footnote{Here, $(a)$ utilizes the condition $ f_c - f_s < \min(B_j) $ for $j \in [N_{\mathrm{sig}}]$ and (b) follows from the lower bound of $f_{j}^{\min}$.}. Thus, the subband extends into and intersects the subsequent intervals.

	The function $\mathcal{F}(\xi)$ maps the row index $\xi$ to a starting frequency of the form $mf_p+kf_c$, where $m\in \{M_1,\ldots,M_2 \}$ and $k\in \{0,\ldots,n-1 \}$. Based on the definition of $\alpha_j$, $f_{j}^{\min}$ falls into the sampling interval of the $\alpha_j$-th row rather than the $(\alpha_j-1)$-th row,                                                                                                                                                                                                                                                                                                                                                                                                                                                                                                                                                                                                                                                                                                                                                                                                                                                                                                                                                                                                                                                                                                                                                                                                                                                                                                                                                                                                                                                                                                                                                                                                                                                                                                                                                                                                                                                                                                                                                                                                                                                                                                                                                                                                                                                                                                                                                                                                                                                                                                                                                                                                                                                                                                                                                                                                                                                                                                                                                                                                                                                                                                                                                                                                                                                                                                                                                                                                                                                                                                                                                                                                                                                                                                                                                                                                                                                                                                                                                                                                                                                                                                                                                                                                                                                                                                                                                                                                                                                                                                                                                                                                                                                                                                                                                                                                                                                                                                                                                                                                                                                                                                                                                                                                                                                                                                                                                                                                                                                                                                                                                                                                                                                                                                                                                                                                                                                                                                                                                                                                                                                                                                                                                                                                                                                                                                                                                                                                                                                                                                                                                                                                                                                                                                                                                                                                                                                                                                                                                                                                                                                                                                                                                                                                                                                                                                                                                                                                                                                                                                                                                                                                                                                                                                                                                                                                                                                                                                                                                                                                                                                                                                                  thereby leading to~\eqref{lemm21}. Similarly, employ the definition of $\beta_j$, we can get~\eqref{lemm22}.
\end{proof}

\section{Proof of Theorem~\ref{theorem_1}}\label{appendix_c}
\begin{figure}[h]
	\centerline{\includegraphics[scale = 0.45]{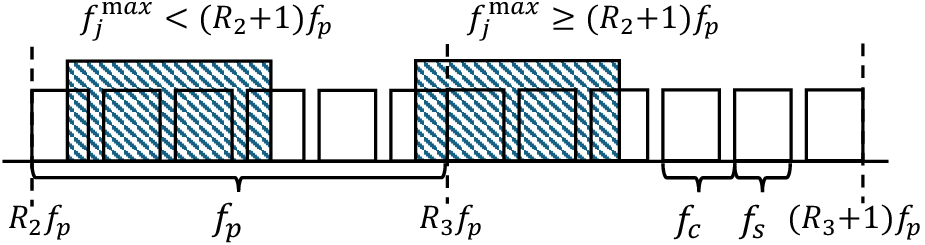}}
	\vspace{-1mm}
	\caption{The schematic diagram of the two cases between subband distribution and sampling intervals.}
	\label{fig:aaa} \vspace{-3mm}
\end{figure}
\begin{proof} Consider a certain subband $X(f)$ for $f\in [ f_{j}^{\min}, f_{j}^{\max}  )$. Let \(R_3\) and \(Z_3\) be integers chosen from the sets \(\{M_1, \ldots, M_2\}\) and \(\{0, \ldots, n-1\}\), respectively. The pair \((R_3, Z_3)\) is selected to maximize the sum \(R_3 f_p + Z_3 f_c\), subject to the constraint \(R_3 f_p + Z_3 f_c \leq f_j^{\min}\). Note the conditions $f_s<f_c<f_p$ and $f_c=\frac{f_p-f_s}{n-1}$, according to Lemma~\ref{lemma_1} and~\ref{lemma_2}, there exist two cases: $f_{j}^{\max}<(R_3+1)f_p$ and $f_{j}^{\max}\geq(R_3+1)f_p$; see Fig.~\ref{fig:aaa}. For the case where $f_{j}^{\max}<(R_3+1)f_p$, $f_{j}^{\max}$ satisfies
	\begin{equation}\label{app_c_1}
		f_{j}^{\max} < R_3f_p+Z_4f_c.
	\end{equation}
	where $Z_4=Z_3+\left \lceil \frac{B_j}{f_c} \right \rceil+1$. Each row of $\mathbf{X}$ corresponds to one sampling interval, and hence the $j$-th subband occupies at most $\left \lceil \frac{B_j}{f_c} \right \rceil+1$ non-zero rows. we have\vspace{-2mm}
	\begin{equation}\label{app_c_6}\vspace{-2mm}
		|\operatorname{supp}(\mathbf{X})| \leq \sum_{j=1}^{N_{\mathrm{sig}}} \lceil B_j/f_c \rceil + N_{\mathrm{sig}}.
	\end{equation}

	For the case $f_{j}^{\max} \geq (R_3+1)f_p$, let $[ f_{j}^{\min}, f_{j}^{\max}  )\subseteq [R_3f_p,(R_4+1)f_p)$ where $R_4$ denotes the smallest possible integer. Divide $[R_3f_p,(R_4+1)f_p)$ into three parts: $[R_3f_p,(R_3+1)f_p)$, $[(R_3+1)f_p,R_4f_p)$, and $[R_4f_p,(R_4+1)f_p)$, with sparsity $n(R_4-R_3-1)$ in the middle part. Let $B_1$ and $B_2$ denote the bandwidths of the first and third parts, respectively. The components of $[R_3f_p,(R_3+1)f_p)$ can be written as
	\begin{equation}
		X(f),f \hspace{-1mm} \in
		\bigcup_{k=n-1-\left \lfloor\frac{B_1}{f_c}\right \rfloor}^{n-1}[R_3f_p+kf_c,R_3f_p+kf_c+f_s ].
	\end{equation}

	For the range $[R_4f_p,(R_4+1)f_p)$, the corresponding components can be obtained as
	\begin{equation}
		X(f), f \hspace{-1mm} \in \hspace{-1mm}\bigcup_{k=0}^{\left \lfloor\frac{B_2}{f_c}\right \rfloor}[R_4f_p+kf_c,R_4f_p+kf_c+f_s].
	\end{equation}

	Combining the contributions from all three parts, the number of non-zero rows in $\mathbf{X}$ of the $j$-th subband is upper-bounded by
	\begin{align}\label{app_c_7}
		 & n(R_4-R_3-1)+\left \lfloor B_1/f_c \right \rfloor+1+\left \lfloor B_2/f_c \right \rfloor+1  \nonumber          \\
		 & \overset{(a)}{<}   n(R_4-R_3-1)+\left \lceil (B_1+B_2)/f_c\right \rceil+1 \nonumber                            \\
		 & =  n(R_4-R_3-1)+\left \lceil (B_j-(R_4-R_3-1)f_p)/f_c \right \rceil+1 \nonumber                                \\
		 & \overset{(b)}{\leq}  n\left \lfloor B_j/f_p\right \rfloor + \lceil \operatorname{mod}(B_j, f_p) / f_c \rceil+1
	\end{align}
	where $(a)$ is due to $\left \lfloor c \right \rfloor+1+\left \lfloor d \right \rfloor+1 < \left \lceil c+d \right \rceil+1$ for $c>0,d > 0$, $(b)$ follows from $R_4-R_3-1 \in \left\{ \left \lfloor\frac{B_j}{f_p}\right \rfloor,\left \lfloor\frac{B_j}{f_p}\right \rfloor-1 \right\}$. Note that when $B_j<f_p$, the sparsity upper bound in~\eqref{app_c_6} coincides with that in~\eqref{app_c_7}. For the overall sampling rate to attain the minimum, we require
	\begin{equation}\label{eq:c_fine}
		2f_s|\operatorname{supp}(\mathbf{X})| + \mathcal{B}+2N_{\mathrm{sig}}f_c \leq 2\mathcal{B}
	\end{equation}
	Combine~\eqref{app_c_6},~\eqref{app_c_7}, and~\eqref{eq:c_fine} and simplifying, we complete the proof of Theorem~\ref{theorem_1}.
\end{proof}

\section{Proof of Lemma~\ref{lemma_3}}\label{appendix_d}
\begin{proof} After arrangement, we have
	\begin{align}\label{eq:lemmm33im}
		 & a(b x+c y)^2 + dy^2 = ab^2x^2 +2abcxy+(ac^2+d)y^2 \nonumber    \\
		 & = ab^2x^2 +2\sqrt{a}b\sqrt{a}cxy+(ac^2+d)y^2 \nonumber         \\
		 & \leq  ab^2x^2 +2\sqrt{a}b\sqrt{ac^2}xy+(ac^2+d)y^2 \nonumber   \\
		 & \leq  ab^2x^2 +2\sqrt{a}b\sqrt{ac^2+d}xy+(ac^2+d)y^2 \nonumber \\
		 & =  (\sqrt{a}bx +\sqrt{ac^2+d}y)^2
	\end{align}
	which completes the proof.
\end{proof}
\section{Proof of Lemma~\ref{lemma_4}}\label{appendix_e}

\begin{proof} Using the properties of the Frobenius norm, it suffices to ensure an arbitrary submatrix (e.g., the $i$-th one) satisfies
	\begin{equation}\label{eq:newwwadqaf}
		\| \mathbf{X}_{S_i}^{(\tilde{S}_{S_i}^{k})^c  }  \|_{F}  \leq \nu _{1} \left \| \mathbf{X}_{S_i}-\mathbf{X}_{S_i}^{k-1} \right \|_{F} + \sqrt{2(1+\delta _{3s})}\left \| \mathbf{E}_{S_i} \right \| _{F} \nonumber
	\end{equation}
	for \eqref{eq:32} to hold. Let $S$ denotes the true support set of an arbitrary submatrix. Define a diagonal matrix $\mathbf{W}_{\Lambda^k}$ with entries $\mathbf{w}_{i,i} = 1$ if $i \in \Lambda^k$, and $\mathbf{w}_{i,i} = \omega$ otherwise, where $\Lambda^k$ denotes the SI set at the $k$-th iteration. Define the operation $\mathcal{R}_{a}(\mathbf{W}_{\Lambda^k},\mathbf{A}_{:,j})\triangleq \mathbf{A}_{:,j} \mathbf{W}_{\Lambda^k} \mathbf{A}^{\dagger}_{:,j} \mathbf{R}^{k - 1}$. For notational simplicity, we omit the subscript $S_i$ in the following proof, with the understanding that the relationship holds for an arbitrary submatrix. From Step 4 of MSSP, we have
	\begin{equation}
		\sum _{j\in S} \| \mathcal{R}_{a}(\mathbf{W}_{\Lambda^k},\mathbf{A}_{:,j}) \|_{F}^{2} \leq  \sum _{j\in \Delta S} \| \mathcal{R}_{a}(\mathbf{W}_{\Lambda^k},\mathbf{A}_{:,j}) \|_{F}^{2}.
	\end{equation}
	By removing the same coordinates $S \cap \Delta S$, we get
	\begin{equation}\label{eq:mergeiq}
		\sum _{j\in S\setminus \Delta S} \hspace{-2mm} \| \mathcal{R}_{a}(\mathbf{W}_{\Lambda^k},\mathbf{A}_{:,j}) \|_{F}^{2}  \leq \hspace{-1.8mm} \sum _{j\in \Delta S \setminus S} \hspace{-2mm}\| \mathcal{R}_{a}(\mathbf{W}_{\Lambda^k},\mathbf{A}_{:,j}) \|_{F}^{2}.
	\end{equation}

	For the right-hand side of~\eqref{eq:mergeiq}, we have
	\begin{align}\label{eq:main_con1}
		 & \sum _{j\in \Delta S \setminus S} \| \mathcal{R}_{a}(\mathbf{W}_{\Lambda^k},\mathbf{A}_{:,j}) \|_{F}^{2} \nonumber                                                                                                                        \\
		 & =  \sum _{j\in \Delta S \setminus S} \| (\mathbf{R}^{k - 1})^H \mathbf{A}^{\dagger}_{:,j} \mathbf{W}_{\Lambda^k}\mathbf{A}_{:,j}^H  \mathbf{A}_{:,j}\mathbf{W}_{\Lambda^k}\mathbf{A}^{\dagger}_{:,j} \mathbf{R}^{k - 1} \|_{F}  \nonumber \\
		 & \leq  \sum _{j\in \Delta S \setminus S} \| \mathbf{W}_{\Lambda^k}\mathbf{A}^{H}_{:,j} \mathbf{R}^{k - 1}  \|_{F}^2  \hspace{-14mm}
	\end{align}
	where the last inequality is because the RIP and the assumption that $\mathbf{A}$ has unit columns so that $\delta_1=0$. Define that $\mathcal{R}_{b}(\mathbf{W}_{\Lambda^k}\mathbf{A}^{H}\mathbf{A})\triangleq(  \mathbf{W}_{\Lambda^k}\mathbf{A}^{H}\mathbf{A}-\mathbf{I})(\mathbf{X}-\mathbf{X}^{k-1})$ and $\mathcal{R}_{c}(\mathbf{E})\triangleq  \mathbf{W}_{\Lambda^k}\mathbf{A}^{H}\mathbf{E}$. By means of the relation $\mathbf{R}^{k - 1}=\mathbf{A}\mathbf{X}-\mathbf{A}\mathbf{X}^{k-1}+\mathbf{E}$, we have
	\begin{align}\label{eq:new_eq_fre}
		 & \hspace{-2mm}  \sum _{j\in \Delta S \setminus S} \| \mathbf{W}_{\Lambda^k}\mathbf{A}^{H}_{:,j} \mathbf{R}^{k - 1}  \|_{F}^2 \nonumber                                   \\
		 & =  \sum _{j\in \Delta S \setminus S} \| \mathbf{W}_{\Lambda^k}\mathbf{A}^{H}_{:,j} (\mathbf{A}\mathbf{X} + \mathbf{E} -\mathbf{A}\mathbf{X}^{k-1} ) \|_{F}^2  \nonumber \\
		 & =   \| \mathbf{W}_{\Lambda^k}\mathbf{A}^{H} (\mathbf{A}\mathbf{X} + \mathbf{E} -\mathbf{A}\mathbf{X}^{k-1} )^{\Delta S \setminus S} \|_{F}^2 \nonumber                  \\
		 & =  \|( \mathcal{R}_{b}(\mathbf{W}_{\Lambda^k}\mathbf{A}^{H}\mathbf{A})+ \mathcal{R}_{c}(\mathbf{E})  )^{\Delta S \setminus S} \|_{F}^2
	\end{align}
	where the last equation is because $(\mathbf{X}-\mathbf{X}^{k-1})^{\Delta S \setminus (S\cup S^{k-1}) }=0$. For the left-hand side of~\eqref{eq:mergeiq}, by the RIP, we have
	\begin{align}\label{eq:main_con2}
		 & \hspace{-2mm} \sum _{j\in S\setminus \Delta S} \hspace{-0.5mm}\| \mathcal{R}_a(\mathbf{W}_{\Lambda^k}, \mathbf{A}_{:,j}  ) \|_{F}^{2} \geq \hspace{-1mm}\sum _{j\in S\setminus \Delta S} \hspace{-0.5mm} \| \mathbf{W}_{\Lambda^k}\mathbf{A}^{H}_{:,j} \mathbf{R}^{k - 1}  \|_{F}^2 \nonumber \\
		 & =   \| \mathbf{W}_{\Lambda^k}\mathbf{A}^{H} (\mathbf{A}\mathbf{X} + \mathbf{E} -\mathbf{A}\mathbf{X}^{k-1} )^{S\setminus \Delta S} \|_{F}^2 \nonumber                                                                                                                                         \\
		 & \overset{}{\geq}  \| \mathbf{W}_{\Lambda^k}\mathbf{A}^{H} (\mathbf{A}\mathbf{X} + \mathbf{E} -\mathbf{A}\mathbf{X}^{k-1} )^{S\setminus (\Delta S\cup S^{k-1})} \|_{F}^2 \nonumber                                                                                                             \\
		 & =  \|( \mathcal{R}_{b}(\mathbf{W}_{\Lambda^k}\mathbf{A}^{H}\mathbf{A})+ \mathcal{R}_{c}(\mathbf{E})  )^{S\setminus \tilde{S}^{k}} + \mathbf{X}^{(\tilde{S}^{k})^c} \|_F^2  \nonumber                                                                                                          \\
		 & \geq  \|\mathbf{X}^{(\tilde{S}^{k})^c} \|_F^2 \hspace{-0.5mm} - \hspace{-0.5mm}\|( \mathcal{R}_{b}(\mathbf{W}_{\Lambda^k}\mathbf{A}^{H}\mathbf{A}) \hspace{-0.5mm}+ \hspace{-0.5mm} \mathcal{R}_{c}(\mathbf{E})  )^{S\setminus \tilde{S}^{k}}\|_{F}^2 \hspace{4mm}
	\end{align}
	where the last equation is due to $(\mathbf{X}-\mathbf{X}^{k-1})^{S\setminus \tilde{S}^{k}} = \mathbf{X}^{ (\tilde{S}^{k})^c}$.

	Combining~\eqref{eq:mergeiq} to~\eqref{eq:main_con2}, we have
	\begin{align}
		 & \hspace{-2mm} \|\mathbf{X}^{(\tilde{S}^{k})^c} \|^2_F \leq  \|( \mathcal{R}_b(\mathbf{W}_{\Lambda^k}\mathbf{A}^{H}\mathbf{A})+ \mathcal{R}_c(\mathbf{E})  )^{\Delta S \setminus S} \|^2_{F}\nonumber \\
		 & +\|( \mathcal{R}_b(\mathbf{W}_{\Lambda^k}\mathbf{A}^{H}\mathbf{A})+ \mathcal{R}_c(\mathbf{E}))^{S\setminus \tilde{S}^{k}}\|^2_F\nonumber                                                             \\
		 & \leq \| (\mathcal{R}_b(\mathbf{W}_{\Lambda^k}\mathbf{A}^{H}\mathbf{A}) +\mathcal{R}_c(\mathbf{E})  )^{ S \cup \tilde{S}^{k}} \|_{F}^2.
	\end{align}

	Taking the square root of both sides and applying the triangle inequality along with \eqref{eq:merge_imp_eq3} and \eqref{eq:lemma7} completes the proof of Lemma~\ref{lemma_4}.
\end{proof}

\section{Proof of Lemma~\ref{lemma_5}}
\begin{proof}
	Since \eqref{eq:19} holds provided that the corresponding inequality is satisfied by any submatrix, we omit the subscript $S_i$ to indicate that the analyses applies to an arbitrary submatrix. Let $S$ denote the true support set of an arbitrary submatrix, since $|T\setminus S|\geq s$, there exists $S'\subseteq T$ with $|S'|=s$ and $S'\cap S=\emptyset$, we conclude that
	\begin{align}\label{eq:imp_eq_1inlema5}
		 & \hspace{-2mm}\sum_{j \in T_{t-s}\cap T_0 } \omega' \| \mathbf{P}^{\perp }_{T\setminus \{j\}}\mathbf{Y} \| ^{2}_{F} + \hspace{-3mm} \sum_{j \in T_{t-s} \setminus  T_0 } \| \mathbf{P}^{\perp }_{T\setminus \{j\}}\mathbf{Y} \| ^{2}_{F} \nonumber \\
		 & \leq  \sum_{j \in S' \cap T_0 } \omega'\| \mathbf{P}^{\perp }_{T\setminus \{j\}}\mathbf{Y} \| ^{2}_{F} + \hspace{-3mm} \sum_{j \in S' \setminus T_0 } \| \mathbf{P}^{\perp }_{T\setminus \{j\}}\mathbf{Y} \| ^{2}_{F}.
	\end{align}

	Define $\mathbf{G} = \mathbf{A}_T^H\mathbf{A}_T$ and let $\hat{\mathbf{X}}^T = \mathbf{G}^{-1}\mathbf{A}_T^H\mathbf{Y}$ be the least-squares estimate on $T$. By partitioned regression theory (see, e.g., \cite{golub2013matrix}, Sec.~6), removing the $j$-th column from $T$ increases the residual sum of squares by
	\begin{equation}\label{eq:downdating}
		\|\mathbf{P}^\perp_{T\setminus\{j\}}\mathbf{Y}\|_F^2
		= \|\mathbf{P}^\perp_T\mathbf{Y}\|_F^2
		+ \frac{\|\hat{\mathbf{X}}^j\|_F^2}{[\mathbf{G}^{-1}]_{jj}},
	\end{equation}
	where $\hat{\mathbf{X}}^j$ denotes the $j$-th row of the least-squares estimate $\hat{\mathbf{X}}^T$, and $[\mathbf{G}^{-1}]_{jj}$ is the $j$-th diagonal element of $\mathbf{G}^{-1}$. This follows from the projection decomposition
	\begin{equation}
		\mathbf{P}^\perp_{T\setminus\{j\}}
		= \mathbf{P}^\perp_T
		+ \frac{\mathbf{P}^\perp_{T\setminus\{j\}}\mathbf{A}_{:,j}\mathbf{A}_{:,j}^H\mathbf{P}^\perp_{T\setminus\{j\}}}{\mathbf{A}_{:,j}^H\mathbf{P}^\perp_{T\setminus\{j\}}\mathbf{A}_{:,j}},
	\end{equation}
	together with $[\mathbf{G}^{-1}]_{jj} = (\mathbf{A}_{:,j}^H\mathbf{P}^\perp_{T\setminus\{j\}}\mathbf{A}_{:,j})^{-1}$ and the relation $\hat{\mathbf{X}}^j = [\mathbf{G}^{-1}]_{jj}\,\mathbf{A}_{:,j}^H\mathbf{P}^\perp_{T\setminus\{j\}}\mathbf{Y}$. Since $\|\mathbf{P}^\perp_T\mathbf{Y}\|_F^2$ is common to all terms, the optimality of $T_{t-s}$ in \eqref{eq:imp_eq_1inlema5} implies
	\begin{align}\label{eq:opt_reduced}
		 & \hspace{-2mm}\sum_{j\in T_{t-s}\cap T_0} \frac{\omega'\|\hat{\mathbf{X}}^j\|_F^2}{[\mathbf{G}^{-1}]_{jj}} + \sum_{j\in T_{t-s} \setminus T_0} \frac{\|\hat{\mathbf{X}}^j\|_F^2}{[\mathbf{G}^{-1}]_{jj}} \nonumber \\
		 & \leq \sum_{j\in S'\cap T_0} \frac{\omega'\|\hat{\mathbf{X}}^j\|_F^2}{[\mathbf{G}^{-1}]_{jj}} + \sum_{j\in S'\setminus T_0} \frac{\|\hat{\mathbf{X}}^j\|_F^2}{[\mathbf{G}^{-1}]_{jj}}
	\end{align}
	Since $|T|=t$ and $\delta_{t}<1$, the RIP guarantees
	\begin{equation}\label{eq:rip_G}
		(1-\delta_{t})\mathbf{I}\preceq \mathbf{G}\preceq (1+\delta_{t})\mathbf{I},
	\end{equation}
	so $\mathbf{G}$ is positive definite and invertible. By \eqref{eq:rip_G}, inverting the spectral bounds gives
	$(1+\delta_{t})^{-1}\mathbf{I}\preceq \mathbf{G}^{-1}\preceq (1-\delta_{t})^{-1}\mathbf{I}$,
	and restricting to diagonal elements yields
	\begin{equation}\label{eq:diag_bound}
		1-\delta_{t} \leq \frac{1}{[\mathbf{G}^{-1}]_{jj}} \leq 1+\delta_{t},
		\quad\forall\, j\in T.
	\end{equation}
	Note that $\omega' = 1+ \omega \geq 1$. Applying the lower bound on the left side and the upper bound on the right side of \eqref{eq:opt_reduced} yielding
	\begin{equation}\label{eq:LS_ineq}
		(1-\delta_{t})\|\hat{\mathbf{X}}^{T_{t-s}}\|_F^2
		\leq \omega'(1+\delta_{t})\|\hat{\mathbf{X}}^{S'}\|_F^2,
	\end{equation}
	where $\hat{\mathbf{X}}^{T_{t-s}}$ and $\hat{\mathbf{X}}^{S'}$ denote the submatrices of $\hat{\mathbf{X}}^T$ restricted to rows in $T_{t-s}$ and $S'$, respectively. Hence
	\begin{equation}\label{eq:hat_bound}
		\|\hat{\mathbf{X}}^{T_{t-s}}\|_F
		\leq \sqrt{\frac{\omega'(1+\delta_{t})}{1-\delta_{t}}}\,\|\hat{\mathbf{X}}^{S'}\|_F.
	\end{equation}

	Since $\mathbf{X}$ has row support $S$, we may write
	$\mathbf{Y} = \mathbf{A}_T\mathbf{X}^T + \tilde{\mathbf{E}}$,
	where $\tilde{\mathbf{E}} \triangleq \mathbf{A}_{ T^c}\mathbf{X}^{T^c} + \mathbf{E}$.
	Substituting into the least-squares estimator and using
	$\mathbf{G}^{-1}\mathbf{A}_T^H\mathbf{A}_T = \mathbf{I}_{|T|}$ gives
	\begin{equation}\label{eq:LS_decomp}
		\hat{\mathbf{X}}^T
		= \mathbf{X}^T + \mathbf{N},
		\qquad
		\mathbf{N} \triangleq \mathbf{A}_T^\dagger\tilde{\mathbf{E}}.
	\end{equation}

	Since $S'\cap S = \emptyset$, we have $\hat{\mathbf{X}}^{S'} = \mathbf{N}^{S'}$.
	On $T_{t-s}$, the triangle inequality gives
	$\|\mathbf{X}^{T_{t-s}}\|_F
		\leq \|\hat{\mathbf{X}}^{T_{t-s}}\|_F + \|\mathbf{N}^{T_{t-s}}\|_F$.
	Applying \eqref{eq:hat_bound} and the trivial bounds
	$\|\mathbf{N}^{S'}\|_F,\,\|\mathbf{N}^{T_{t-s}}\|_F\leq\|\mathbf{N}\|_F$, we obtain
	\begin{equation}\label{eq:combined}
		\|\mathbf{X}^{T_{t-s}}\|_F
		\leq \left(\sqrt{\frac{\omega'(1+\delta_{t})}{1-\delta_{t}}}+1\right)\|\mathbf{N}\|_F
		.
	\end{equation}
	For the term $\mathbf{N}$,
	$\|\mathbf{A}_T^\dagger\|_{2\to 2}
		\leq 1/\sqrt{1-\delta_{t}}$
	by \cite[Lemma~6]{CIte47}.
	Moreover, applying the RIP column-wise to
	$\mathbf{A}_{S \setminus T}$ with $|S \setminus T|\leq s$ gives
	$\|\mathbf{A}_{T^c}\mathbf{X}^{T^c}\|_F = \|\mathbf{A}_{S\setminus T}\mathbf{X}^{S\setminus T}\|_F
		\leq \sqrt{1+\delta_s}\,\|\mathbf{X}^{T^c}\|_F$.
	Combining these with the triangle inequality
	$\|\tilde{\mathbf{E}}\|_F
		\leq \|\mathbf{A}_{T^c}\mathbf{X}^{T^c}\|_F
		+ \|\mathbf{E}\|_F$, we obtain
	\begin{align}\label{eq:final}
		\|\mathbf{X}^{T_{t-s}}\|_F
		 & \leq \left(\frac{\sqrt{\omega'(1+\delta_{t})}+\sqrt{1-\delta_{t}}}{1-\delta_{t}}\right) \nonumber \\
		 & \quad\quad \cdot\left(\sqrt{1+\delta_s}\,\|\mathbf{X}^{ T^c}\|_F
		+ \|\mathbf{E}\|_F\right)
	\end{align}
	which completes the proof.
\end{proof}

\section{Proof of Lemma~\ref{lemma_6}}
\begin{proof}
	Let $S = U \cup V$, so that $|S| \leq s + t$. Denote
	$\mathbf{M} = ((\mathbf{I}_L - \mathbf{W}_{T_0}\mathbf{A}^H\mathbf{A})\mathbf{V})^U$.
	Since $\mathbf{M}$ is supported on $U \subseteq S$ and $\mathbf{V}$ is
	supported on $V \subseteq S$, the analysis can be restricted to the
	subspace indexed by $S$. Let $\mathbf{W}_S \in \mathbb{R}^{|S|\times|S|}$
	denote the principal submatrix of $\mathbf{W}_{T_0}$ indexed by $S$,
	i.e., a diagonal matrix with diagonal entries $1$ for
	$i \in S\cap T_0$ and $\omega$ for $i \in S\setminus T_0$.
	By the support structure, we have
	\begin{align}
		\|\mathbf{M}\|_F
		 & = \|(\mathbf{I}_{|S|} - \mathbf{W}_S\mathbf{A}_S^H\mathbf{A}_S)
		\mathbf{V}_S\|_F \nonumber                                                     \\
		 & \leq \|\mathbf{I}_{|S|} - \mathbf{W}_S\mathbf{A}_S^H\mathbf{A}_S\|_{2\to 2}
		\,\|\mathbf{V}\|_F.
	\end{align}
	Writing
	$\mathbf{I}_{|S|} - \mathbf{W}_S\mathbf{A}_S^H\mathbf{A}_S
		= (\mathbf{I}_{|S|} - \mathbf{W}_S)
		+ \mathbf{W}_S(\mathbf{I}_{|S|} - \mathbf{A}_S^H\mathbf{A}_S)$
	and applying the triangle inequality gives
	\begin{align}
		 & \|\mathbf{I}_{|S|} - \mathbf{W}_S\mathbf{A}_S^H\mathbf{A}_S\|_{2\to 2}
		\nonumber                                                                 \\
		 & \leq \|\mathbf{I}_{|S|} - \mathbf{W}_S\|_{2\to 2}
		+ \|\mathbf{W}_S\|_{2\to 2}\,
		\|\mathbf{I}_{|S|} - \mathbf{A}_S^H\mathbf{A}_S\|_{2\to 2}. \nonumber
	\end{align}
	Since $\mathbf{W}_S$ is diagonal with entries in $[\omega, 1]$,
	we have $\|\mathbf{I}_{|S|} - \mathbf{W}_S\|_{2\to 2} = 1 - \omega$
	and $\|\mathbf{W}_S\|_{2\to 2} = 1$.
	Moreover, the RIP gives
	$\|\mathbf{I}_{|S|} - \mathbf{A}_S^H\mathbf{A}_S\|_{2\to 2}
		\leq \delta_{s+t}$. Therefore
	\begin{equation}
		\|\mathbf{M}\|_F
		\leq (1 - \omega + \delta_{s+t})\,\|\mathbf{V}\|_F,
	\end{equation}
	which completes the proof.
\end{proof}
\section{Proof of Lemma~\ref{lemma_7}}
\begin{proof}
	The lemma follows trivially from the fact that
	\begin{align}
		 & \|\left(\mathbf{W}_{T_{0}}\mathbf{A}^H \mathbf{E}\right)^{ U}\|_F^2 = | \langle \mathbf{W}_{T_{0}}\mathbf{A}^H \mathbf{E},\left(\mathbf{W}_{T_{0}}\mathbf{A}^H \mathbf{E}\right)^{ U}\rangle | \nonumber \\
		 & =  | \langle\mathbf{E}, \mathbf{A}\mathbf{W}_{T_{0}}\left(\mathbf{W}_{T_{0}}\mathbf{A}^H \mathbf{E}\right)^{ U}\rangle  | \nonumber                                                                      \\
		 & \leq \left\|\mathbf{E}\right\|_F\|\mathbf{A}\mathbf{W}_{T_{0}}\left(\mathbf{W}_{T_{0}}\mathbf{A}^H \mathbf{E}\right)^{ U}\|_F \nonumber                                                                  \\
		 & \leq \sqrt{1+\delta_u} \left\|\mathbf{E}\right\|_F \| \left(\mathbf{W}_{T_{0}}\mathbf{A}^H \mathbf{E}\right)^{ U}\|_F
	\end{align}
	where the last inequality is due to the RIP and $\omega \leq 1$. Square $\|\left(\mathbf{W}_{T_{0}}\mathbf{A}^H \mathbf{E}\right)^{ U}\|_F$ on both sides, Lemma~\ref{lemma_7} is proven.
\end{proof}
\bibliographystyle{IEEEtran}
\bibliography{Mybib}

\end{document}